\documentclass{article}
\usepackage{amsmath,amsfonts}
\usepackage{gensymb}
\usepackage{amssymb}
\usepackage[left=1in, right=1in, top=1in, bottom=1in]{geometry}
\usepackage{graphicx}
\usepackage{tikz}
\usetikzlibrary{positioning}
\usepackage{mathtools}
\usepackage{amsthm}
\usetikzlibrary{arrows,automata,calc,math}
\usepackage{nccmath}
\usepackage{xspace}
\usepackage{xcolor}
\usepackage{booktabs}
\usepackage{array}
\definecolor{MyBlue}{rgb}{0.12, 0.12, 0.76}
\usepackage[colorlinks,allcolors=MyBlue]{hyperref}
\usepackage{pifont}

\newcommand\tat{t\^{a}tonnement\xspace}

\usepackage[numbers]{natbib}

\makeatletter
\newcommand{\thickhline}{%
    \noalign {\ifnum 0=`}\fi \hrule height 1.4pt
    \futurelet \reserved@a \@xhline
}
\newcolumntype{"}{@{\hskip\tabcolsep\vrule width 1.4pt\hskip\tabcolsep}}
\makeatother

\usepackage{thmtools}
\usepackage{thm-restate}

\newtheorem{theorem}{Theorem}[section]
\newtheorem{lemma}{Lemma}[section]

\newtheorem{corollary}{Corollary}[theorem]

    \newtheoremstyle{TheoremNum}
        {\topsep}{\topsep}              
        {\itshape}                      
        {}                              
        {\bfseries}                     
        {.}                             
        { }                             
        {\thmname{#1}\thmnote{ \bfseries #3}}
    \theoremstyle{TheoremNum}
    \newtheorem{theoremNumbered}{Theorem}

\usepackage{algorithm}
\usepackage{algpseudocode}
\usepackage{algorithmicx}
\let\oldReturn\Return
\renewcommand{\Return}{\State\oldReturn}
\algtext*{EndWhile}
\algtext*{EndIf}
\algtext*{EndForAll}
\algtext*{EndFor}
\algtext*{EndFunction}

\DeclareMathOperator*{\argmax}{arg\,max}

\newcommand\dif{\mathop{}\!\mathrm{d}}
\allowdisplaybreaks 

\newcommand\bbr{\mathbb{R}}
\newcommand\bbrpos{\mathbb{R}_{\ge 0}}
\newcommand\bbrspos{\mathbb{R}_{> 0}}
\newcommand\bbnspos{\mathbb{N}_{>0}}

\newcommand\ep{\varepsilon}
\newcommand{\A}{\mathbf{a}}
\newcommand{\x}{\mathbf{x}}
\newcommand\X{\mathcal{X}}
\newcommand{\y}{\mathbf{y}}

\newcommand{\p}{\mathbf{p}}
\newcommand\q{\mathbf{q}}

\newcommand\xprime{\mathbf{x'}}
\newcommand\yprime{\mathbf{y'}}

\newcommand\zero{\mathbf{0}}
\newcommand\one{\mathbf{1}}

\newcommand\B{\mathbf{b}}
\newcommand\bigB{\mathbf{B}}
\newcommand\bfeta{\boldsymbol{\eta}}
\newcommand\phat{\hat{p_i}}
\newcommand\E{\mathcal{E}}
\newcommand\bfa{\boldsymbol{\alpha}}

\newcommand\bflam{\boldsymbol{\lambda}}
\newcommand\bflamp{\boldsymbol{\lambda'}}
\newcommand\tilp{\tilde{p_i}}

\usepackage{enumitem}
\setlist{itemsep=0pt} 

\begin{document}

\title{Counteracting Inequality in Markets via Convex Pricing}

\author{Ashish Goel \and Benjamin Plaut}

\date{Stanford University\\ \texttt{\{ashishg,\ bplaut\}@stanford.edu}}

\maketitle

\pagestyle{plain}

\begin{abstract}
We study market mechanisms for allocating divisible goods to competing agents with quasilinear utilities. For \emph{linear} pricing (i.e., the cost of a good is proportional to the quantity purchased), the First Welfare Theorem states that Walrasian equilibria maximize the sum of agent valuations. This ensures efficiency, but can lead to extreme inequality across individuals. Many real-world markets -- especially for water -- use \emph{convex} pricing instead, often known as increasing block tariffs (IBTs). IBTs are thought to promote equality, but there is a dearth of theoretical support for this claim.

\ \ \ \ In this paper, we study a simple convex pricing rule and show that the resulting equilibria are guaranteed to maximize a CES welfare function. Furthermore, a parameter of the pricing rule directly determines which CES welfare function is implemented; by tweaking this parameter, the social planner can precisely control the tradeoff between equality and efficiency. Our result holds for any valuations that are homogeneous, differentiable, and concave. We also give an iterative algorithm for computing these pricing rules, derive a truthful mechanism for the case of a single good, and discuss Sybil attacks.

\end{abstract}

\section{Introduction}\label{sec:intro}

Markets are one of the oldest mechanisms for distributing resources; indeed, commodity prices were meticulously recorded in ancient Babylon for over 300 years~\cite{Spek2005,Spek2003}. In a market, buyers and sellers exchange goods according to some sort of pricing system, and \emph{Walrasian equilibrium}\footnote{This is also known as market equilibrium, competitive equilibrium, and general equilibrium, depending on the context.} occurs when the demand of the buyers exactly equals the supply of the sellers. This concept was first studied by Walras in the 1870's~\cite{Walras1874}. In 1954, Arrow and Debreu showed that under some conditions, a Walrasian equilibrium is guaranteed to exist~\cite{Arrow1954}. Most of the literature on Walrasian equilibrium only considers \emph{linear} pricing, meaning the cost of a good is proportional to the quantity purchased.

In this paper, we consider the problem of allocating divisible goods to competing agents via a market mechanism. We assume each agent has quasilinear utility: an agent's utility is her value for the resources she obtains (her \emph{valuation}), minus the money she spends (her \emph{payment}). The First Welfare Theorem states that in this setting, the linear-pricing Walrasian equilibria are exactly the allocations maximizing utilitarian welfare, i.e., the sum of agent valuations. Thus linear pricing \emph{implements} utilitarian welfare in Walrasian equilibrium (sometimes abbreviated ``WE").


The result is powerful, but also limiting. Maximizing utilitarian welfare yields the most efficient outcome, but may also cause maximal inequality (see Figure~\ref{fig:linear}). 

One common alternative is \emph{convex} pricing. In this paper, we study convex pricing rules $p$ of the form
\begin{align*}
p(x_i) = \Big(\sum_{j} q_j x_{ij}\Big)^{1/\rho}
\end{align*}
where $x_i$ is \emph{bundle} agent $i$ receives, $x_{ij} \in \bbrpos$ is the fraction of good $j$ she receives, $q_1,\dots,q_m$ are constants, and $\rho \in (0,1]$ determines the curvature of the pricing rule. Like linear pricing, $p$ is still \emph{anonymous}, meaning that agents' payments depend only on their purchases (and not on their preferences, for example).

When $\rho  =1$, $p$ reduces to linear pricing. When $\rho < 1$, $p$ is strictly convex, meaning that doubling one's consumption will more than double the price. This will make it easy to buy a small amount, but hard to buy a large amount, which intuitively should lead to a more equal distribution of resources. As the curvature of the pricing rule grows, this effect should be amplified, leading to a different equality/efficiency tradeoff.

Our work seeks to formalize that claim. We will show that the Walrasian equilibria of these convex pricing rules are guaranteed to maximize a \emph{constant elasticity of substitution} (CES) welfare function, where the choice of $\rho$ determines the specific welfare function and thus the precise equality/efficiency tradeoff (Theorem~\ref{thm:main}). Our result holds for a wide range of agent valuations.

\newcommand{\stickman}[4]{
    \def\xx{#1} 
    \def\yy{#2} 
    \def\rr{#3} 
    \def\name{#4} 
    \def\xl{\xx-1.2*\rr} 
    \def\xr{\xx + 1.2*\rr} 
    \def\ya{\yy-\rr} 
    \def\yb{\ya-.5*\rr} 
    \def\yc{\yb-1.5*\rr} 
    \def\ybd{\yb-1.2*\rr} 
    \def\ycd{\yc-1.8*\rr} 

    \draw[
        thick
            ]
    (\xx, \yy) circle (\rr) 
    (\xx, \yy) node at (\xx, \yy + 2*\rr) {\name}
    (\xx, \ya) -- (\xx,\yc) 
    (\xl, \ybd) -- (\xx, \yb) -- (\xr, \ybd) 
    (\xl, \ycd) -- (\xx, \yc) -- (\xr, \ycd) 
    ; %
} 

\newcommand\glass[4]{
	\def\xtl{#1} 
	\def\yt{#2} 
	\def\sc{#3} 
	\def\full{#4} 
	\def\xbl{\xtl + .13*\sc}
	\def\yb{\yt - \sc}
	\def\xbr{\xbl + .4*\sc}
	\def\xtr{\xbr + .13*\sc}
	\def\ym{\yt - .15*\sc}
	\def\theangle{25}
	\definecolor{columbiablue}{rgb}{0.61, 0.87, 1.0}
	\definecolor{capri}{rgb}{0.0, 0.75, 1.0}	
		
	\begin{scope}
	\clip (\xtl, \yt) -- (\xbl, \yb) to[bend right=\theangle] (\xbr, \yb) -- (\xtr, \ym) to[bend right=\theangle] (\xtl, \ym);
	\ifthenelse{\full=1}{\fill[capri] (\xtl, \yb-1) rectangle (\xtr,\ym+1)}{};	
	\end{scope}

	\begin{scope}
	\clip (\xtl, \ym) to[bend right=\theangle] (\xtr, \ym) to[bend right=\theangle] (\xtl, \ym);
	\ifthenelse{\full=1}{\fill[columbiablue] (\xtl, \yb-1) rectangle (\xtr,\ym+1)}{};	
	\end{scope}

	\draw [semithick]
		(\xtl, \yt) -- (\xbl, \yb)
		(\xbr, \yb) -- (\xtr, \yt)
		(\xtl, \yt) to[bend right=\theangle] (\xtr, \yt)			
		;
	\draw[line width=.3pt]	
		(\xtl, \yt) to[bend left=\theangle] (\xtr, \yt)
		;
	\draw[semithick] (\xbl, \yb) to[bend right=\theangle] (\xbr, \yb);
		\ifthenelse{\full=0}{\draw[line width=.3pt](\xbl, \yb) to[bend left=\theangle] (\xbr, \yb)}{};
}

\begin{figure}[tb]
\begin{center}
\begin{tikzpicture}
\newcommand\stickspace{1.7}
\newcommand\largestickspace{2}
\newcommand\secondstart{4.3*\stickspace}
\newcommand\glassspace{.25*\stickspace}

\stickman{0}{0}{0.22} {$w_1 = 1$}
\stickman{\stickspace}{0}{0.22} {$w_2 = 6$}
\stickman{2*\stickspace}{0}{0.22} {$w_3 = 5$}
\path[thick,->] (2*\stickspace + 3*\rr, -2*\rr) edge node[above] {price $=6$} (\secondstart - 3*\rr, -2*\rr);

\stickman{\secondstart}{0}{0.22} {$w_1 = 1$}
\glass{\secondstart + \glassspace}{-.12}{.5}{0}

\stickman{\secondstart + \largestickspace}{0}{0.22} {$w_2 = 6$}
\glass{\secondstart + \largestickspace + \glassspace}{-.12}{.5}{1}

\stickman{\secondstart + 2*\largestickspace}{0}{0.22} {$w_3 = 5$}
\glass{\secondstart + 2*\largestickspace + \glassspace}{-.12}{.5}{0}

\end{tikzpicture}
\end{center}
\caption{An example of how linear pricing can lead to maximal inequality. Consider the three agents above and a single good (say, water), where each agent $i$'s value for $x$ units of the good is $w_i \cdot x$. The unique linear-pricing Walrasian equilibrium sets a price of 6 per unit, which results in agent 2 buying all of the good and the other two agents receiving nothing. More generally, the equilibrium price reflects the maximum anyone is willing to pay, and anyone who is not willing to pay that much is priced out of the market and receives nothing. In contrast, our nonlinear pricing rule always ensures that everyone receives a nonzero amount; see Section~\ref{sec:results}.}
\label{fig:linear}
\end{figure}
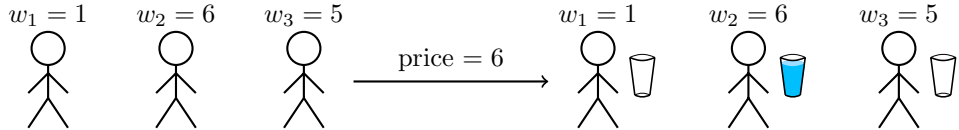

\paragraph{Convex pricing in the real world.} Convex pricing is especially pervasive in the water sector, where such pricing rules are known as \emph{increasing block tariffs} (IBTs)~\cite{Whittington1992}, typically implemented with discrete blocks of water (hence the name). IBTs have been implemented and empirically studied in Israel~\cite{Becker2015}, South Africa~\cite{Burger2014}, Spain~\cite{Garcia-Rubio2015}, Jordan~\cite{Klassert2018}, and the United States~\cite{Renwick2000}, among many other countries. 

IBTs are often claimed to promote equality in water access~\cite{Whittington1992}, but there has been limited theoretical evidence supporting this (see~\cite{Monteiro2011} for one of the only examples). On the other hand, a common concern is that IBTs may lead to poor ``economic efficiency"~\cite{Boland2000,Monteiro2011}. Our work shows that at least on a theoretical level, convexity of pricing does not necessarily lead to inefficiency: it simply maximizes a different welfare function than the traditional utilitarian one. In particular, it maximizes a CES welfare function.


\paragraph{The Second Welfare Theorem and personalized pricing.} The Second Welfare Theorem is perhaps the most famous theoretical result regarding implementation in Walrasian equilibrium. It states any Pareto optimum can be a WE when an arbitrary redistribution of initial wealth is allowed.\footnote{Specifically, for any Pareto optimal allocation, there exists a redistribution of initial wealth which makes that allocation a WE. However, our quasilinear utility model does not have a concept of initial wealth (alternatively, initial wealth is simply an additive constant in agents' utilities which does not affect their behavior), so this result is not as mathematically relevant. See Section~\ref{sec:prior} for additional discussion.} Another method that achieves the same goal is \emph{personalized pricing}, where different agents can be charged different (linear) prices. In contrast, convex pricing is anonymous: agents purchasing the same bundle always pay the same price.

Each of these approaches certainly has its own pros and cons. In this paper, our goal is not to claim that convex pricing is ``better" than other approaches (or vice versa). Regardless of which is ``better" in any given situation, convex pricing \emph{is} widely used in practice, and is often claimed to promote equality. Our goal in this paper is to formally quantify that claim.


\subsection{CES welfare functions}\label{sec:ces-intro}

A welfare function~\cite{Bergson1938,Samuelson1947} assigns a real number to each possible outcome, with higher numbers (i.e, higher welfare) indicating outcomes that are more desirable to the social planner. Different welfare functions represent different priorities; our focus will be the tradeoff between overall efficiency and individual equality. For a fixed constant $\rho \in (-\infty, 0) \cup (0,1]$, the constant elasticity of substitution (CES) welfare of outcome $\x$ is
\[
\Phi(\rho, \x) = \bigg(\sum_{\text{agents } i} v_i(\x)^{\rho}\bigg)^{1/\rho}
\]
where $v_i(\x)$ is agent $i$'s value for $\x$. In general, different values of $\rho$ will lead to different optimal allocations, so whenever we say ``maximum CES welfare allocation", we mean with respect to a fixed value of $\rho$. See Section~\ref{sec:model} for an axiomatic characterization of CES welfare functions.

For $\rho = 1$, this is utilitarian welfare, i.e., the sum of valuations. The limit as $\rho \to -\infty$ yields max-min welfare (the minimum valuation)~\cite{Rawls2009,Sen1977,Sen1976}, whereas $\rho \to 0$ yields Nash welfare (the product of valuations)~\cite{Kaneko1979,Nash1950}. The closer $\rho$ gets to $-\infty$, the more the social planner cares about individual equality (max-min welfare being the extreme case of this), and the closer $\rho$ gets to 1, the more the social planner cares about overall societal good (utilitarian welfare being the extreme case of this). For this reason, $\rho$ is called the \emph{inequality aversion} parameter, and this family of welfare functions is thought to exhibit an \emph{equality/efficiency} tradeoff.

These welfare functions were originally proposed by Atkinson~\cite{Atkinson70}; indeed, his motivation was to measure the level of inequality in a society. Despite being extremely influential in the traditional economics literature (see~\cite{Cowell2011} for a survey), the CES welfare function has received almost no attention in the computational economics community.\footnote{To our knowledge, only three other computational economics papers have studied CES welfare in any context:~\cite{Arunachaleswaran2019,Goel2019nash,Plaut2019}.}

Finally, note that $\Phi(\rho, \x)$ is defined with respect to the each agent's valuation $v_i$ and \emph{not} her overall quasilinear utility $u_i$. We acknowledge that it is standard to define welfare with respect to the overall utility $u_i$, and we have two reasons for not doing so. First, in the case of scarce resources, a social planner may be interested in equality in \emph{consumption} (e.g., equality in water access), not just equality in utility derived. Second, it turns out mathematically that this is the welfare function maximized by convex pricing WE in our model; the version where $\Phi(\rho, \x)$ considers $u_i$ may be not be maximized by the resulting WE. This may yield valuable qualitative insights about convex pricing; for example, does convex pricing lead to equality with respect to consumption but not necessarily with respect to underlying utility?

\paragraph{CES welfare in healthcare.} These welfare functions have also seen substantial use in healthcare under the name of \emph{isoelastic welfare functions}. This began with~\cite{Wagstaff1991}, largely motivated by concerns abeout purely utilitarian approaches to healthcare (i.e., allocating resources to maximize total health in a community, without concern for equality). Since these decisions can affect who lives and who dies, significant effort has been invested into understanding the equality/efficiency tradeoff, with this class of welfare functions serving as a theoretical tool~\cite{Dolan1998,Ortun1996,Wagstaff1991}; see Section~\ref{sec:prior} for additional discussion. Despite the ongoing interest in this tradeoff, the healthcare literature has not (to our knowledge) considered convex pricing as a mechanism for balancing equality and efficiency.

More broadly, our work can be thought of as weaving together the previously disjoint threads of CES welfare and convex pricing to provide theoretical support for the oft-cited but rarely quantified claim that IBTs promote equality.


\section{Results and related work}\label{sec:results}

\paragraph{Main result: convex pricing implements CES welfare maximization in Walrasian equilibrium.} Our main result is that for convex pricing of the form $p(x_i) = (\sum_j q_j x_{ij})^{1/\rho}$ for any $\rho \in (0,1]$\footnote{The case of $\rho < 0$ is slightly unintuitive, as it can result in agents who care more receiving \emph{less} of the good. Consequently, implementation in WE is impossible; see Theorem~\ref{thm:neg-rho-counter}.}, a Walrasian equilibrium is guaranteed to exist, and every WE maximizes CES welfare with respect to $\rho$. This holds for a wide range of agent valuations.

\begin{theoremNumbered}[\ref{thm:main} \textnormal{(Simplified version)}] 
Assume each valuation is homogeneous of degree $r$,\footnote{A valuation is homogeneous of degree $r$ if scaling any bundle by a constant $c$ scales the resulting value by $c^r$.} differentiable, and concave, and fix $\rho \in (0,1]$. Then an allocation $\x = (x_1,\dots, x_n)$ maximizes CES welfare if and only if there exist constants $q_1,\dots,q_m \in \bbrpos$ such that for the pricing rule
\[
p(x_i) = \Big(\sum_{j} q_j x_{ij}\Big)^{1/\rho},
\]
$\x$ and $p$ form a WE.
\end{theoremNumbered}

\noindent Note that the $\rho$ in $p(x_i)$ is the same $\rho$ for which CES welfare is maximized.

We call the reader's attention to two important aspects of this result. Perhaps most importantly, our result is not simply a reformulation of the First Welfare Theorem: although maximizing CES welfare for valuations $v_1,\dots,v_n$ is equivalent to maximizing utilitarian welfare for valuations $v_1^\rho,\dots,v_n^\rho$, the First Welfare Theorem does not say anything about the agent demands in response to this convex pricing rule. The First Welfare Theorem also does not help with identifying the exact conditions under which Theorem~\ref{thm:main} holds, e.g., homogeneity of valuations.\footnote{In fact, not only is homogeneity necessary, but homogeneity of the same degree is necessary: if we allow the degree of homogeneity to differ across agents, the result no longer holds (Theorem~\ref{thm:different-r}).}


Secondly, the class of homogeneous, differentiable, and concave valuations is quite large: it generalizes most of the commonly studied valuations, e.g., linear, Cobb-Douglas, and CES (note that here we are referring to CES agent valuations, not CES welfare functions). Although Leontief valuations are not differentiable, we handle them as a special case and show that the same result holds (Theorem~\ref{thm:leontief}).


The following additional properties are of note:

\begin{enumerate}
\item For this class of utilities, Theorem~\ref{thm:main} generalizes the First Welfare Theorem:\footnote{One direction of the First Welfare Theorem (if $(\x, p)$ is a linear pricing WE, then $\x$ maximizes utilitarian welfare) holds in a much more general setting; see Appendix~\ref{sec:linear}.} when $\rho = 1$, $p(x_i)$ yields linear pricing and CES welfare yields utilitarian welfare.
\item The constants $q_1,\dots,q_m$ will be the optimal Lagrange multipliers for a convex program maximizing CES welfare. This connection to duality will be very helpful for computing these WE (see Section~\ref{sec:comp}).
\item Our pricing rule is strictly convex for $\rho < 1$, with the curvature growing as $\rho$ goes to 0. The smaller $\rho$ gets, the easier it is to buy a small amount, but the harder it is to buy a large amount. Intuitively, this should prevent any single individual from dominating the market and lead to a more equitable outcome. Furthermore, the marginal price at $x_i = \zero$ is zero, which ensures that everyone ends up with a nonempty bundle (in contrast to linear pricing: see Figure~\ref{fig:linear}). Theorem~\ref{thm:main} provides a tight relationship between the curvature of the pricing rule and the exact equality/efficiency tradeoff.
\end{enumerate}

\paragraph{Towards an implementation.} We also prove several supporting results: in particular, regarding implementation. The WE from Theorem~\ref{thm:main} can always be computed by asking each agent for her entire utility function, and then solving a convex program for maximizing CES welfare maximization to obtain the optimal Lagrange multipliers $q_1,\dots,q_m$. However, this is not very practical: people are generally not able to articulate a full cardinal utility function, and even if they are, doing so could require transmitting an enormous amount of information. Section~\ref{sec:comp} presents our first supporting result: an iterative algorithm for computing the WE, where in each step, each agent only needs to report the gradient of her valuation at the current point. Our algorithm is based on the ellipsoid method, and inherits its polynomial-time convergence properties. We recognize that even valuation gradient queries may be difficult for agents to answer, and we leave the possibility of an improved implementation -- in particular, a \emph{\tat}\footnote{A \tat is an iterative algorithm which only asks \emph{demand queries}, i.e., what would each agent purchase given the current prices. Demand queries may be easier to answer than the valuation gradient queries in our algorithm.} -- as an open question. 

\paragraph{Truthfulness.} Our second supporting result considers a different approach to implementation: \emph{truthful} mechanisms. Walrasian equilibria are generally not truthful: agents can lie about their preferences to affect the equilibrium prices for their personal gain.\footnote{Another interpretation is that WE assumes agents are \emph{price-taking} (i.e., treat the prices are given and do not lie about their preferences to affect the equilibrium prices) and breaks down when agents are \emph{price-anticipating}.} For $\rho = 1$, the Vickrey-Clarke-Groves (VCG) mechanism is known to truthfully maximize utilitarian welfare~\cite{Nisan2007}. For the case of a single good and any $\rho \in (0,1)$, we give a mechanism which truthfully maximizes CES welfare (Theorem~\ref{thm:truthful}). We also show that our mechanism is the unique truthful mechanism up to an additive constant in the payment rule (Theorem~\ref{thm:truthful-unique}). The proof of Theorem~\ref{thm:truthful-unique} is quite involved, and requires techniques from real analysis such as Kirszbraun's Theorem for Lipschitz extensions and the Fundamental Theorem of Lebesgue Calculus.



\paragraph{Negative results.} We prove the following negative results. Most importantly, we show that for any $\rho \ne 1$, linear-pricing WE can have arbitrarily poor CES welfare (Theorem~\ref{thm:rho=1-bad}); were this not the case, perhaps it would suffice to simply use linear pricing and accept an approximation of CES welfare. Next, note that Theorem~\ref{thm:main} requires each agent's valuation to be homogeneous with the same degree $r$. We show that when agents' valuations have different homogeneity degrees, there exist instances where no pricing rule can implement CES welfare maximization in WE (Theorem~\ref{thm:different-r}), and thus our assumption is necessary. We also show that CES welfare maximization cannot be implemented in WE for $\rho < 0$ (Theorem~\ref{thm:neg-rho-counter}), and discuss the special case of $\rho = 0$ (i.e., Nash welfare).

There is an additional crucial issue which any practical implementation of Theorem~\ref{thm:main} would need to address: \emph{Sybil attacks}. A Sybil attack is when a selfish agent attempts to gain an advantage in a system by creating fake identities~\cite{Douceur2002}. Since the pricing rule from Theorem~\ref{thm:main} is strictly convex for $\rho < 1$, an agent can decrease her payment by masquerading as multiple individuals and splitting her purchase across those identities.\footnote{In contrast, for $\rho = 1$, there is nothing to be gained by creating fake identities.} In Appendix~\ref{sec:sybil}, we propose a model for analyzing Sybil attacks in markets, and show that if these attacks are possible, there exist instances where no pricing rule can implement CES welfare maximization in WE (Theorem~\ref{thm:sybil-other-p}).\footnote{There are combinations of parameters, however, where our pricing rule is naturally robust to Sybil attacks: in particular, when $v_i(\x) (1-\rho) \le \kappa$ (where $v_i(\x)$ is agent $i$'s value for the maximum CES welfare allocation and $\kappa$ is the identity creation cost). This suggests a natural way for an equality-focused social planner to choose a specific value for $\rho$: estimate the identity creation cost and scale of valuations in the system of interest, and pick $\rho$ to be as small as possible without incentivizing Sybil attacks.}

\paragraph{Additional results.} In Appendix~\ref{sec:fisher}, we explore connections between our results in the quasilinear utility model, and the Fisher market fixed-budget model. Appendix~\ref{sec:leontief} shows that Theorem~\ref{thm:main} extends to Leontief valuations, which are not differentiable (so the main proof does not apply). Leontief valuations have been a focus of prior work, so we find is worthwhile to handle this as a special case.

\subsection{Related work}\label{sec:prior}

The study of markets has a long history in economics~\cite{Arrow1954,Brainard2005,Fisher1892,Varian1974,Walras1874}. Recently, this topic has received substantial attention in the computer science community as well (see~\cite{Vazirani2007} for an algorithmic introduction). We first provide some important background on different market models and the First and Second Welfare Theorems, and then move on to more recent related work.

\paragraph{Quasilinear markets and Fisher markets.} There are two primary market models for divisible goods. This paper considers the quasilinear utility model, where each agent can spend as much as she wants, and the amount spent is incorporated into her utility function. The other predominant model is the \emph{Fisher market} model~\cite{Brainard2005,Fisher1892}, where each agent has a fixed budget constraint, and the amount spent does not affect her resulting utility (as a result, each agent always spends exactly her budget). Although these two models share many of the same conceptual messages, some of the technical results vary. See Appendix~\ref{sec:fisher} for the technical relationship between the two models with respect to WE and CES welfare maximization.

Since agents in Fisher markets always spend exactly their budgets, there is no way to elicit the absolute scale of agent valuations. Nash welfare maximization is invariant to this type of scaling, but no other CES welfare function is~\cite{Moulin2003}. For this reason, the Fisher market model is not well suited to reason about other welfare functions. In contrast, the quasilinear model \emph{does} allow agents to express the absolute scale of their valuation: specifically, by choosing how much to spend. That is one reason that we focus on the quasilinear model for this paper. The other is that convex pricing is most easily applied to a small submarket of the broader economy (e.g., water pricing), and quasilinear utility captures the fact that agents may wish to spend money on other goods outside of this submarket. In contrast, Arrow and Debreu's model (see below) can arguably capture the entire economy, so there is nothing outside of the market to spend money on.

\paragraph{The First and Second Welfare Theorems.} Conceptually, the First Welfare Theorem establishes an efficiency property that any WE must satisfy, and the Second Welfare Theorem deals with implementing a wide range of allocations as WE. The two welfare theorems originate in the context of \emph{Arrow-Debreu} markets~\cite{Arrow1954}, which generalize Fisher markets to allow for (1) agents to enter the market with goods (as opposed to just money)\footnote{These are known as ``exchange markets" or ``exchange economies".} and (2) production of goods. The statements of the First and Second Welfare Theorems in that model are, respectively, ``Any (linear pricing) WE is Pareto optimal" and ``Any Pareto optimal allocation can be a (linear pricing) WE with \emph{transfers}, i.e., under a suitable redistribution of initial wealth". 

In the Fisher market and quasilinear utility models, the First Welfare Theorem can be strengthened to ``Any (linear pricing) WE maximizes budget-weighted Nash welfare"~\cite{Eisenberg1961,Eisenberg1959,Vazirani2007} and ``Any (linear pricing) WE maximizes utilitarian welfare", respectively. The version of the Second Welfare Theorem stated above is appropriate for Fisher markets, since agents' budgets constitute the ``initial wealth". However, for quasilinear utilities, there is no notion of initial wealth (alternatively, initial wealth is an additive constant in agents' utilities which does not affect their behavior). Thus for quasilinear utilities, allowing transfers actually does not affect the set of WE. This may seem counterintuitive, since the Second Welfare Theorem (which still holds in this setting) states that any Pareto optimum can be a WE. However, Pareto optimality here is referring to agents' overall quasilinear utilities, \emph{not} the agents' valuations. It can be shown that the only allocations which are Pareto optimal with respect to the quasilinear utilities are allocations maximizing utilitarian welfare, which are already covered by the First Welfare Theorem (without transfers). 

Thus on a technical level, the Second Welfare Theorem is not helpful in the world of quasilinear utilities. However, even when the Second Welfare Theorem is mathematically relevant, a centrally mandated redistribution of wealth is often out of the question in practice.

\paragraph{The equality/efficiency tradeoff in healthcare.}
In Section~\ref{sec:ces-intro}, we discussed how CES welfare has been studied from a theoretical perspective in healthcare~\cite{Dolan1998,Ortun1996,Wagstaff1991}. There have also been several empirical studies aiming to understand the general population's view of the equality/efficiency tradeoff, with results generally indicating a disapproval of purely utilitarian approaches to healthcare~\cite{Dolan2000,Wittrup-Jensen2008}. For example, a survey of 449 Swedish politicians found widespread rejection of purely utilitarian decision-making in healthcare, and under some conditions, the respondents were willing to sacrifice up to 15 of 100 preventable deaths in order to ensure equality across subgroups~\cite{Lindholm1998}.

\paragraph{CES welfare and $\alpha$-fairness in networking.} CES welfare functions have also enjoyed considerable attention from the field of networking, under the name of \emph{$\alpha$-fairness} (the parameter $\alpha$ corresponds to $1-\rho$ in our definition). The $\alpha$-fairness notion was proposed by~\cite{Mo2000}, motivated in part as a generalization of the prominent \emph{proportional fairness} objective (which is equivalent to Nash welfare)~\cite{Kelly1998}. See~\cite{Bertsimas2012} and references therein for further background on $\alpha$-fairness in networking. To our knowledge, a market-based understanding was developed only for proportional fairness, starting with the seminal work of Kelly et al.~\cite{Kelly1998}.

\paragraph{Nonlinear market mechanisms and CES welfare maximization.} We are aware of just two papers studying market mechanisms for CES welfare functions: \cite{Goel2019nash} and \cite{Plaut2019}. Like our work, both of these papers explore nonlinear pricing rules, but unlike our work, only consider Leontief valuations. Furthermore, both of those papers are in the Fisher market model and only achieve CES welfare maximization under strong assumptions on the absolute scale of the agents' valuations.\footnote{In particular, that each agent's weight for each good is either 0 or 1. This subclass of Leontief valuations is known as ``bandwidth allocation" valuations, where each good is a link in a network, and agents transmit data over fixed paths.} In contrast, our main result holds for any valuations that are homogeneous of degree $r$, differentiable, and concave, a much larger range of valuations. (Leontief valuations are not differentiable, but we handle them as a special case in Appendix~\ref{sec:leontief} and show that our result still holds.) It is worth noting that~\cite{Goel2019nash} focuses on the WE model, whereas~\cite{Plaut2019} considers strategic agents and Nash equilibria. On a related note, we are not aware of any broader results regarding general nonlinear pricing, i.e., what set of allocations can be implemented if we allow $p(x_i)$ to be any nondecreasing function of $x_i$ (but still require anonymity)? This could be an interesting direction for future work.

The rest of the paper is organized as follows. Section~\ref{sec:model} formally defines the model. In Section~\ref{sec:main}, we present our main result: a simple convex pricing rule implements CES welfare maximization in WE for $\rho \in (0,1]$ (Theorem~\ref{thm:main}). Section~\ref{sec:comp} presents an iterative algorithm for computing these WE. In Section~\ref{sec:truthful}, we consider truthful mechanisms for CES welfare maximization. Section~\ref{sec:sybil} discusses Sybil attacks, and Section~\ref{sec:counter} presents our negative results. At this point we conclude the main paper. Appendix~\ref{sec:fisher} discusses connections to Fisher markets, Appendix~\ref{sec:leontief} shows that our main result extends to Leontief valuations, Appendix~\ref{sec:linear} discusses the First Welfare Theorem in more detail, and Appendix~\ref{sec:proofs} provides some proofs omitted from the main paper.

\section{Model}\label{sec:model}

Let $N = \{1,2,\ldots n\}$ be a set of agents, and let $M = \{1,2,\dots m\}$ be a set of divisible goods. Throughout the paper, we use $i$ and $k$ to refer to agents and $j$ and $\ell$ to refer to goods. We need to determine an \emph{allocation} $\x \in \bbr_{\ge 0}^{n\times m}$, where $x_i \in \bbr_{\ge 0}^m$ is the \emph{bundle} of agent $i$, and $x_{ij} \in [0,1]$\footnote{Without loss of generality, we can normalize the supply of each good to be 1.} is the quantity of good $j$ allocated to agent $i$. An allocation cannot allocate more than the available supply: $\x$ is a valid allocation if and only if $\sum_i x_{ij} \leq 1$ for all $j$. We will also determine \emph{payments} $p_1,\dots,p_n$, where $p_i$ is the payment charged to agent $i$. Thus an outcome is an allocation $\x$ and payments $p_1,\dots,p_n$.

Agent $i$'s utility for a bundle $x_i$ is denoted by $u_i(x_i)  \in \bbr$. We assume that agents have quasilinear utility: for an allocation $\x$ and payment $p_i$, agent $i$'s utility is $u_i(x_i) = v_i(x_i) - p_i$, where $v_i$ is agent $i$'s \emph{valuation}. When each agent's payment only depends on the bundle she receives, i.e., $p_i = p(x_i)$, we call $p$ a \emph{pricing rule}. With the exception of Section~\ref{sec:truthful}, we will focus on pricing rules. For $v_i$, we make the following standard assumptions throughout the paper:
\begin{enumerate}
\item Nonzero: There exists a bundle $x_i$ such that $v_i(x_i) > 0$.
\item Montone: If $x_{ij} \ge y_{ij}$ for all $j \in M$, then $v_i(x_i) \ge v_i(y_i)$.
\item Normalized: $v_i(0,\dots, 0) = 0$.
\end{enumerate}

Our positive results require the following three additional properties, which we will mention explicitly whenever used:
\begin{enumerate}
\setcounter{enumi}{3}
\item Concave: For any bundles $x_i, y_i$ and constant $\lambda \in [0,1]$, we have $v_i(\lambda x_i + (1 - \lambda) y_i) \geq \lambda v_i( x_i) + (1-\lambda)u_i(y_i)$.
\item Homogeneous of degree $r$: for any bundle $x_i$ and constant $\lambda \ge 0$, $v_i(\lambda x_i) = \lambda^r v_i(x_i)$. For $0 < r < 1$, this models diminishing returns. Note that homogeneity implies normalization, and for monotone and concave $v_i$, we must have $r \ge 0$ and $r \le 1$ respectively. 
\item Differentiable: for any bundle $x_i$ and all $j \in M$, $\mfrac{\partial v_i(x_i)}{\partial x_{ij}}$ is defined.
\end{enumerate}


\paragraph{CES welfare.} For multipliers $\A = (a_1, a_2\dots a_n) \in \bbrpos^n$ and $\rho \in (-\infty, 0)\cup(0,1]$, the (weighted) CES welfare of an allocation $\x$ is $\Phi_\A(\rho, \x) = \big(\sum_{i \in N} a_i v_i(x_i)^{\rho}\big)^{1/\rho}$. For $\rho \ne 1$, $\Phi$ is strictly concave in $v_i(x_i)$ for all $i \in N$, so every optimal allocation $\x$ has the same valuation vector $v_1(x_1),\dots, v_n(x_n)$. We will use $\Psi_\A(\rho)$ to denote CES welfare maximization, i.e., $\Psi_\A(\rho) = \argmax_{\x \in \bbrpos^m:\ \sum_i x_{ij} \le 1\ \forall j} \Phi_\A(\rho, \x)$. There may be multiple optimal allocations (for example, if there is a good which no one values), so $\Psi_\A(\rho)$ denotes a set.  Thus $\x \in \Psi_\A(\rho)$ denotes that $\x$ has maximum CES welfare. When each agent has the same multiplier (other than Appendix~\ref{sec:fisher}, this will always be the case), we simply write $\Phi(\rho, \x)$ and $\Psi(\rho)$.

As an illustrative example, consider a single good and valuations that are homogeneous of degree 1. Utilitarian welfare results in the good being entirely allocated to agents with $w_i = \max_k w_k$, with other agents receiving nothing (see Figure~\ref{fig:linear}). In contrast, for $\rho < 1$, the unique allocation maximum CES welfare welfare gives the following bundle $x_i \in \bbrspos$ to each agent $i$ (Lemma~\ref{lem:m=1-x}): $x_i = \mfrac{{w_i}^{\frac{\rho}{1-\rho}}}{\sum_{k} {w_k}^{\frac{\rho}{1-\rho}}}$. One natural case is $\rho = 1/2$, which results in a proportional allocation.

\paragraph{Axiomatic characterization.} CES welfare functions also admit an axiomatic characterization. Consider the following axioms: (1) Monotonicity: if one agent's valuation increases while all others are unchanged, the welfare function should prefer the new allocation, (2) Anonymity: the welfare function should treat all agents the same, (3) Continuity: the welfare function should be continuous.\footnote{A slightly weaker version of continuity is often used: if an allocation $\x$ is strictly preferred to an allocation $\y$, there should be neighborhoods $N(\x)$ and $N(\y)$ such that every $\xprime \in N(\x)$ is preferred to every $\yprime \in N(\y)$. This weaker version only requires a welfare \emph{ordering} and does not require that this ordering be expressed by a function. However, any such ordering which also satisfies the rest of our axioms is indeed representable by a welfare function~\cite{Debreu1960}, and so both sets of axioms end up specifying the same set of welfare functions/orderings.}, (4) Independence of common scale: scaling all agent valuations by the same factor should not affect which allocations have better welfare than others, (5) Independence of unconcerned agents: when comparing the welfare of two allocations, the comparison should not depend on agents who have the same valuation in both allocations, and (6) The Pigou-Dalton principle: when choosing between equally efficient allocations, the welfare function should prefer more equitable allocations~\cite{Dalton1920,Pigou1912}. 

Disregarding monotonic transformations of the welfare function (which of course do not affect which allocations have better welfare than others), the set of welfare functions satisfying these axioms is exactly the set of CES welfare functions with $\rho \in (-\infty,0)\cup(0,1]$, including Nash welfare~\cite{Moulin2003}.\footnote{This actually does not include max-min welfare, which satisfies weak monotonicity but not strict monotonicity.} This axiomatic characterization shows that we are not just focusing on an arbitrary class of welfare functions: CES welfare functions are arguably the most reasonable welfare functions.

\paragraph{Walrasian equilibrium.} Given a pricing rule $p$, agent $i$'s \emph{demand set} is defined by $D_i(p) = \argmax_{x_i \in \bbrpos^m} u_i(x_i)$, or equivalently, $D_i(p) = \argmax_{x_i \in \bbrpos^m} \big(v_i(x_i) - p(x_i)\big)$. Given an allocation $\x$ and payment rule $p$, $(\x, p)$ is a\emph{Walrasian equilibrium} (WE) if both of the following hold:
\begin{enumerate}
\item Each agent receives a bundle in her demand set: $x_i \in D_i(p)$ for all $i \in N$.
\item The market clears: for all $j \in M$, $\sum_{i \in N} x_{ij} \le 1$, and for all $j \in M$ with nonzero cost, $\sum_{i \in N} x_{ij} = 1$.\footnote{We say that good $j$ has nonzero cost for $j$ if there is a bundle $x_i$ such that $x_{i\ell} = 0$ for all $\ell \ne j$, but $p(x_i) > 0$.}
\end{enumerate}

\section{Main result}\label{sec:main}

We begin with our main result: for a wide range of valuations and any $\rho \in (0,1]$, a simple convex pricing rule leads to CES welfare maximization in Walrasian equilibrium. Our pricing rule has many additional interesting properties; to avoid redundancy, we refer the reader back to our discussion in Section~\ref{sec:results}. On a high level, the proof relies on the KKT conditions for CES welfare maximization and the KKT conditions for each agent's demand set, and uses Euler's Theorem for homogeneous functions to conjoin the two. This will result in the following theorem:

\begin{restatable}{theorem}{mainThm}\label{thm:main}
Assume each $v_i$ is homogeneous of degree $r$, concave, and differentiable. For any $\rho \in (0, 1]$ and any allocation $\x$, we have $\x \in \Psi(\rho)$ if and only if there exist $q_1,\dots, q_m\in \bbrpos$ such that for the pricing rule
\[
p(x_i) = \rho r^{\frac{\rho -1}{\rho}} \Big(\sum_{j \in M} q_j x_{ij}\Big)^{1/\rho},
\]
$(\x, p)$ is a WE. Furthermore, $q_1,\dots, q_m$ are optimal Lagrange multipliers for Program~\ref{pro:ces}.
\end{restatable}

\subsection{Proof setup}
We begin by setting up the two relevant convex programs and proving several lemmas. For valuations $v_1\dots v_n$, nonnegative multipliers $\A = a_1\dots a_n$, and $\rho \in (-\infty, 0) \cup(0,1]$, consider the following nonlinear program for maximizing CES welfare:
\begin{alignat}{2}
\max\limits_{\x \in \bbrpos^{n\times m}} &\ \frac{1}{\rho}\sum_{i \in N} a_i v_i(x_i)^\rho \label{pro:ces} \\ 
s.t.\ &\ \sum\limits_{i \in N} x_{ij}\leq 1\quad &&\ \forall j \in M  \nonumber
\end{alignat}

Since the constraints are linear and the objective function is concave (since $\rho \le 1$), Program~\ref{pro:ces} is a convex program. Program~\ref{pro:ces} depends on $\rho$, but we will leave this implicit when clear from context: we will simply say ``$\x$ is optimal for Program~\ref{pro:ces}" as opposed to ``$\x$ is optimal for Program~\ref{pro:ces} with respect to $\rho$". Note also that we are maximizing $\frac{1}{\rho}\sum_{i \in N} a_i v_i(x_i)^\rho$ instead of the true CES welfare $\Phi_\A(\rho, \x) = (\sum_{i \in N} a_i v_i(x_i)^\rho)^{1/\rho}$; this will lead to the same optimal allocation $\x$ and will simplify the analysis.

When $\A$ is not specified, we assume that $\A = \one$. Nonuniform multipliers will only be used in Appendix~\ref{sec:fisher} when we consider connections to Fisher markets, but we include them here for completeness. 

%

Next, consider each agent's demand set given a pricing rule $p$:
\begin{align}
D_i(p) = \argmax_{x_i \in \bbrpos^m}\ \big(v_i(x_i) - p(x_i)\big) \label{pro:demand}
\end{align}
When $p$ is convex (as in Theorem~\ref{thm:main}), $-p$ is concave. Since $v_i$ is also concave, $v_i(x_i) - p(x_i)$ is concave, so each agent's demand set defines a convex program (Program~\ref{pro:demand}). Program~\ref{pro:demand} depends on $i$, the agent in question, but again we leave this implicit when it is clear from context.

We will also use the following theorem, due to Euler. We include a short proof in Appendix~\ref{sec:proofs}.\footnote{The reason we provide a proof is that this theorem is often stated with the requirement of continuous differentiability, but in fact only requires differentiability; to avoid any confusion, we provide a proof only using differentiability.}

\begin{restatable}[Euler's Theorem for homogeneous functions]{theorem}{thmEuler}\label{thm:euler}
Let $f: \bbrpos^m \to \bbr$ be differentiable and homogenous of degree $r$. Then for any $\B = (b_1,\dots b_m) \in \bbrpos^m$, $\sum_{j = 1}^m b_j \mfrac{\partial f(\B)}{\partial b_j} = r f(\B)$.
\end{restatable}

Before we state and prove Theorem~\ref{thm:main}, we note one other property: for a pricing rule of the form $p(x_i) = c(\sum_{j \in M} q_j x_{ij})^{1/\rho}$ where $c>0$, good $j$ has nonzero cost (for the purposes of Walrasian equilibrium) if and only if $q_j = 0$.

\subsection{Proof of Theorem~\ref{thm:main}}

The proof of Theorem~\ref{thm:main} is divided into three parts. The first part involves setting up the KKT conditions for Programs~\ref{pro:ces} and \ref{pro:demand}. The second assumes that $\x \in \Psi(\rho)$ and proves that $(\x, p)$ is a WE, and the third assumes that $(\x, p)$ is a WE and proves that $\x \in \Psi(\rho)$.


\begin{proof}[Proof of Theorem~\ref{thm:main}]

\textbf{Part 1: Setup.} Let $\q$ denote the vector $(q_1,\dots,q_m) \in \bbrpos^m$; then the Lagrangian of Program~\ref{pro:ces} is $L(\x, \q) = \frac{1}{\rho} \sum_{i \in N} v_i(x_i)^\rho - \sum_{j \in M} q_j (\sum_{i \in N} x_{ij} - 1)$.\footnote{The expert reader may notice that we have omitted the $\x \in \bbrpos^{m\times n}$ constraint from the Lagrangian. We do this to slightly simplify the analysis. The effect on the KKT conditions is that stationarity changes from ``For all $i,j$, $\frac{\partial L(\x, \q)}{\partial x_{ij}} = 0$" to ``For all $i,j$, $\frac{\partial L(\x, \q)}{\partial x_{ij}} \le 0$, and the inequality holds with equality when $x_{ij} > 0$".} Since Program~\ref{pro:ces} is convex and satisfies strong duality by Slater's condition, the KKT conditions are both necessary and sufficient for optimality. That is, $\x$ is optimal for Program~\ref{pro:ces} (which is equivalent to $\x \in \Psi(\rho)$) if and only if there exist Lagrange multipliers $\q \in \bbrpos^m$ such that both of the following hold:\footnote{The KKT conditions also include primal feasibility and dual feasibility. Since we will only work with valid allocations $\x$ and nonnegative $q_1,\dots,q_m$, these two conditions are trivially satisfied.}
\begin{enumerate}
\item Stationarity: $\mfrac{\partial L(\x, \q)}{\partial x_{ij}} \le 0$ for all $i,j$. Furthermore, if $x_{ij} > 0$, the inequality holds with equality.
\item Complementary slackness: for all $j \in M$, either $\sum_{i \in N} x_{ij} = 1$, or $q_j = 0$.
\end{enumerate}

For a given $(i,j)$ pair, $\mfrac{\partial L(\x, \q)}{\partial x_{ij}}$ is equal to $v_i(x_i)^{\rho - 1} \mfrac{\partial v_i(x_i)}{\partial x_{ij}} - q_j$, so stationarity for Program~\ref{pro:ces} is equivalent to: $q_j \ge v_i(x_i)^{\rho - 1} \mfrac{\partial v_i(x_i)}{\partial x_{ij}}$ for all $i,j$, and when $x_{ij} > 0$, the inequality holds with equality.

Next consider Program~\ref{pro:demand}, which defines each agent's demand set. This program has no constraints (other than $x_i \in \bbrpos^m$), so we can ignore complementary slackness. Thus by the KKT conditions, $x_i \in D_i(p)$ if and only if for every $j \in M$,  $\mfrac{\partial v_i(x_i)}{\partial x_{ij}} \le \mfrac{\partial p(x_i)}{\partial x_{ij}}$, and if $x_{ij} > 0$, the inequality holds with equality (stationarity). We can explicitly compute the partial derivatives of $p$: $\mfrac{\partial p(x_i)}{\partial x_{ij}} = r^{\frac{\rho -1}{\rho}} q_j \big(\sum_{\ell \in M} q_\ell x_{i\ell}\big)^{\frac{1-\rho}{\rho}}$.

\textbf{Part 2: Optimal CES welfare implies WE.} Suppose that $\x \in \Psi(\rho)$. Then there exists $\q \in \bbrpos^m$ such that $q_j \ge v_i(x_i)^{\rho - 1} \mfrac{\partial v_i(x_i)}{\partial x_{ij}}$ for all $j$, and $q_j = v_i(x_i)^{\rho - 1} \mfrac{\partial v_i(x_i)}{\partial x_{ij}}$ whenever $x_{ij} > 0$. Using the latter in combination with Euler's Theorem for homogeneous functions, for each $(i,j)$ pair we have
\begin{align*}
\frac{\partial p(x_i)}{\partial x_{ij}} =&\ r^{\frac{\rho -1}{\rho}} q_j \Big(\sum_{\ell \in M} q_\ell x_{i\ell}\Big)^{\frac{1-\rho}{\rho}}\\
=&\ q_j \Big(r^{-1}  \sum_{\ell: x_{i\ell} > 0} q_\ell x_{i\ell}\Big)^{\frac{1-\rho}{\rho}}\\
=&\ q_j \Big(r^{-1}  \sum_{\ell: x_{i\ell} > 0} v_i(x_i)^{\rho - 1} \frac{\partial v_i(x_i)}{\partial x_{i\ell}} x_{i\ell}\Big)^{\frac{1-\rho}{\rho}} &&\ \text{(stationarity for $x_{i\ell}$ when $x_{i\ell} > 0$)}\\
=&\ q_j \Big(r^{-1}  v_i(x_i)^{\rho - 1} \sum_{\ell \in M}  \frac{\partial v_i(x_i)}{\partial x_{i\ell}} x_{i\ell}\Big)^{\frac{1-\rho}{\rho}}\\
=&\ q_j \big(r^{-1} v_i(x_i)^{\rho - 1} r v_i(x_i)\big)^{\frac{1-\rho}{\rho}} \quad\quad &&\ \text{(Euler's Theorem)}\\
=&\ q_j v_i(x_i)^{1-\rho}
\end{align*}
Thus $\mfrac{\partial p(x_i)}{\partial x_{ij}} = q_j v_i(x_i)^{1-\rho}$. Next, we claim that $x_i \in D_i(p)$ for all $i \in N$. Fix an agent $i$; we show by case analysis that $x_i$ satisfies stationarity (for Program~\ref{pro:demand}) for each $j \in M$.

Case 1: $q_j = v_i(x_i)^{\rho - 1} \mfrac{\partial v_i(x_i)}{\partial x_{ij}}$. Then $\mfrac{\partial p(x_i)}{\partial x_{ij}} = q_j v_i(x_i)^{1-\rho} = v_i(x_i)^{\rho - 1} \mfrac{\partial v_i(x_i)}{\partial x_{ij}} v_i(x_i)^{1-\rho} = \mfrac{\partial v_i(x_i)}{\partial x_{ij}}$,
and we are done.

Case 2: $x_{ij} = 0$ and $q_j \ge v_i(x_i)^{\rho - 1} \mfrac{\partial v_i(x_i)}{\partial x_{ij}}$. Then similarly, $\mfrac{\partial p(x_i)}{\partial x_{ij}} = q_j v_i(x_i)^{1-\rho}
\ge v_i(x_i)^{\rho - 1} \mfrac{\partial v_i(x_i)}{\partial x_{ij}} v_i(x_i)^{1-\rho}
= \mfrac{\partial v_i(x_i)}{\partial x_{ij}}$,
and again we are done. Therefore $x_i \in D_i(p)$ for all $i \in N$. 

Since $\x$ is a valid allocation, $\sum_{i \in N} x_{ij} \le 1$ for all $j \in M$. This, combined with complementary slackness for Program~\ref{pro:ces}, is identical to the market clearing condition for Walrasian equilibrium. Thus we have shown that if $\x \in \Psi(\rho)$, there exist $q_1,\dots, q_m$ (which are optimal Lagrange multipliers for Program~\ref{pro:ces}) such that for pricing rule $p$ as defined, $(\x, p)$ is a WE.

\textbf{Part 3: WE implies optimal CES welfare.} This is similar to Part 2. Suppose there exists $\q \in \bbrpos^m$ such that for pricing rule $p(x_i) = \rho r^{\frac{\rho -1}{\rho}} (\sum_{j \in M} q_j x_{ij})^{1/\rho}$, $(\x, p)$ is a WE. Recall the partial derivatives of $p$: $\mfrac{\partial p(x_i)}{\partial x_{ij}} = q_j  \big(r^{-1} \sum_{\ell \in M} q_\ell x_{i\ell}\big)^{\frac{1-\rho}{\rho}}$. We multiply each side by $x_{ij} r^{-1}$, and sum both sides over $j$:
\begin{align*}
r^{-1} \sum_{j \in M} x_{ij} \frac{\partial p(x_i)}{\partial x_{ij}} =&\ \Big(r^{-1}\sum_{\ell \in M} q_\ell x_{i\ell}\Big)^{\frac{1-\rho}{\rho}} r^{-1} \sum_{j \in M} q_j x_{ij}\\
r^{-1}\sum_{j: x_{ij} > 0} x_{ij} \frac{\partial p(x_i)}{\partial x_{ij}} =&\ \Big(r^{-1}\sum_{\ell \in M} q_\ell x_{i\ell}\Big)^{\frac{1-\rho}{\rho} + 1}
\end{align*}
Since $(\x, p)$ is a WE, we have $x_i \in D_i(p)$ for all $i \in N$. Thus $\mfrac{\partial v_i(x_i)}{\partial x_{ij}} = \mfrac{\partial p(x_i)}{\partial x_{ij}}$ whenever $x_{ij} > 0$, so
\begin{align*}
r^{-1}\sum_{j: x_{ij} > 0} x_{ij} \frac{\partial v_i(x_i)}{\partial x_{ij}} =&\ \Big(r^{-1}\sum_{\ell \in M} q_\ell x_{i\ell}\Big)^{1/\rho}\\
r^{-1}\sum_{j \in M} x_{ij} \frac{\partial v_i(x_i)}{\partial x_{ij}} =&\ \Big(r^{-1}\sum_{\ell \in M} q_\ell x_{i\ell}\Big)^{1/\rho}\\
v_i(x_i) =&\ \Big(r^{-1}\sum_{\ell \in M} q_\ell x_{i\ell}\Big)^{1/\rho} \quad\quad \text{(Euler's Theorem)}\\
v_i(x_i)^{\rho-1} =&\ \Big(r^{-1}\sum_{\ell \in M} q_\ell x_{i\ell}\Big)^{\frac{\rho-1}{\rho}}
\end{align*}
Using this in combination with $\mfrac{\partial p(x_i)}{\partial x_{ij}} = q_j  \big(r^{-1} \sum_{\ell \in M} q_\ell x_{i\ell}\big)^{\frac{1-\rho}{\rho}}$, we get
\[
q_j = \frac{\partial p(x_i)}{\partial x_{ij}}\Big(r^{-1}\sum_{\ell \in M} q_\ell x_{i\ell}\Big)^{\frac{\rho-1}{\rho}}
= \frac{\partial p(x_i)}{\partial x_{ij}} v_i(x_i)^{\rho - 1}
\]

Next, we claim that $(\x, \q)$ satisfies stationarity for Program~\ref{pro:ces}. We proceed by case analysis for each $(i,j)$ pair. Stationarity for Program~\ref{pro:demand} implies that these are the only two possible cases.

Case 1: $\mfrac{\partial v_i(x_i)}{\partial x_{ij}} = \mfrac{\partial p(x_i)}{\partial x_{ij}}$. In this case, $q_j = \mfrac{\partial v_i(x_i)}{\partial x_{ij}} v_i(x_i)^{\rho - 1}$, and we are done. 

Case 2: $x_{ij} = 0$ and $\mfrac{\partial v_i(x_i)}{\partial x_{ij}} \le \mfrac{\partial p(x_i)}{\partial x_{ij}}$. In this case we have $q_j \ge \mfrac{\partial v_i(x_i)}{\partial x_{ij}} v_i(x_i)^{\rho - 1}$, and we are again done. 

Thus $(\x, \q)$ satisfies stationarity for Program~\ref{pro:ces}. Furthermore, the second condition of Walrasian equilibrium is again identical to the complementary slackness condition. We conclude that $\x \in \Psi(\rho)$, and that $q_1,\dots, q_m$ are optimal Lagrange multipliers for Program~\ref{pro:ces}. This completes the proof.
\end{proof}

The following corollary states that under this pricing rule, each agent's resulting value will be proportional to the her payment. This property will be helpful in future sections, and may also be interesting independently. 

\begin{restatable}{corollary}{corUtilityPrice}\label{cor:utility-price}
Assume each $v_i$ is homogeneous of degree $r$, concave, and differentiable, and let $p(x_i) = (\sum_{j \in M} q_j x_{ij})^{1/\rho}$ for some $\q\in\bbrpos^m$. Then if $x_i \in D_i(p)$, $p(x_i) = \rho r v_i(x_i)$.
\end{restatable}

\begin{proof}
As before, stationarity for Program~\ref{pro:demand} gives us $\mfrac{\partial v_i(x_i)}{\partial x_{ij}} = \mfrac{\partial p(x_i)}{\partial x_{ij}}$ whenever $x_{ij} > 0$. Also note that by definition, $p$ is homogeneous of degree $1/\rho$. Using these two properties in combination with Euler's Theorem, we get
\begin{align*}
\frac{\partial v_i(x_i)}{\partial x_{ij}} =&\ \frac{\partial p(x_i)}{\partial x_{ij}} \quad \text{for all $j \in M$ s.t. $x_{ij} > 0$} \\
\sum_{j \in M} x_{ij} \frac{\partial v_i(x_i)}{\partial x_{ij}} =&\ \sum_{j \in M} x_{ij} \frac{\partial p(x_i)}{\partial x_{ij}}\\
r v_i(x_i) =&\ \frac{1}{\rho} p(x_i)
\end{align*}
Multiplying both sides by $1/\rho$ completes the proof.
\end{proof}

\section{Towards an implementation}\label{sec:comp}

Theorem~\ref{thm:main} guarantees the existence of Walrasian equilibria maximizing CES welfare, but says nothing about how to find these equilibria. As discussed in Section~\ref{sec:results}, we could always explicitly ask each agent for her valuation, and directly solve Program~\ref{pro:ces}. However, agents are generally not able to articulate their entire valuations, and even if they are, doing so could be extremely tedious.

In this section, we give an iterative algorithm for computing the WE given by Theorem~\ref{thm:main}. The algorithm will just compute the optimal allocation; Lemma~\ref{lem:x-to-q} shows how the equilibrium pricing rule can easily be obtained once the optimal allocation is in hand. Our algorithm is computationally equivalent to running the general-purpose ellipsoid method on Program~\ref{pro:ces}, i.e., it explores the exact same sequence of allocations. The key is that we are able to implement the ellipsoid method only using valuation gradient queries, i.e., ``tell me the gradient of your valuation at this point". We immediately inherit the correctness and polynomial-time convergence properties of the ellipsoid algorithm. Throughout this section, we make the same assumptions as in Theorem~\ref{thm:main}: each $v_i$ is concave, homogeneous of degree $r$, and differentiable.

First, recall that the pricing rule from Theorem~\ref{thm:main} takes the form $p(x_i) = \rho r^{\frac{\rho - 1}{\rho}} (\sum_{j \in M} q_j x_{ij})^{1/\rho}$. Since $\rho$ and $r$ are constants, it suffices to compute $\q = q_1,\dots,q_m$. Helpfully, Theorem~\ref{thm:main} tells us that if $\q$ are optimal Lagrange multipliers for Program~\ref{pro:ces}, then $(\x, p)$ is a WE for any $\x \in \Psi(\rho)$. The next lemma states if we know an $\x \in \Psi(\rho)$, and have access to the gradients of the agents' valuations at $\x$, we can determine optimal Lagrange multipliers.

\begin{restatable}{lemma}{lemXToQ}\label{lem:x-to-q}
Let $\x \in \Psi(\rho)$. Then we can determine optimal Lagrange multipliers $\q$ using only\\ $\nabla v_1(x_1),\dots,\nabla v_n(x_n)$.
\end{restatable}

\begin{proof}
First, using Euler's Theorem for homogeneous functions (Theorem~\ref{thm:euler}), we can obtain $v_1(x_1),\dots,v_n(x_n)$ using only $\nabla v_1(x_1),\dots,\nabla v_n(x_n)$. Next, since $\x \in \Psi(\rho)$, the KKT conditions for Program~\ref{pro:ces} imply that whenever $x_{ij} > 0$, $q_j = v_i(x_i)^{\rho - 1} \mfrac{\partial v_i(x_i)}{\partial x_{ij}}$. Fix a $j \in M$. If $x_{ij} > 0$ for some agent $i$, then $q_j = v_i(x_i)^{\rho - 1} \mfrac{\partial v_i(x_i)}{\partial x_{ij}}$. We know all the values on the right hand side, so we can compute $q_j$. if $x_{ij} = 0$ for all $i \in N$, then complementary slackness implies that $q_j = 0$.
\end{proof}


Thus it suffices to find an allocation $\x \in \Psi(\rho)$, which is equivalent to finding an $\x$ that is optimal for Program~\ref{pro:ces}. There are many iterative algorithms for solving convex programs of this form. Furthermore, many only require (1) oracle access to the objective function and its gradient, and (2) \emph{a separation oracle}\footnote{A separation oracle is an algorithm which, given a point $x$ and a convex set $\mathcal{X}$, determines whether $x \in \mathcal{X}$. If $x\not\in\mathcal{X}$, it must return a separating hyperplane (if $\mathcal{X}$ is specified by a set of constraints, returning a violated constraint is sufficient).} for the constraint set (and no additional assumptions of strong convexity or other properties). For the sake of specificity, we focus on the \emph{ellipsoid method}~\cite{Bubeck2015}, but any algorithm with these properties is sufficient for our purposes.

\begin{lemma}[\cite{Bubeck2015}]\label{lem:converge}
Let $f$ be a convex function and let $\mathcal{X}$ be a convex set. Consider the program $\min_{x \in \mathcal{X}} f(x)$.
Let $\E$ be a ball containing the minimum of $f$, and suppose there exists a polynomial-time separation oracle for $\mathcal{X}$. Then the ellipsoid method starting from $\E$ requires only oracle access to $f$ and $\nabla f$, and converges to the minimum of $f$ in polynomial time.
\end{lemma}


In our case, we have a trivial polynomial-time separation oracle: simply check each constraint to see if it is violated. For the gradient of our objective function, we have $\mfrac{\partial}{\partial x_{ij}} \Big(\frac{1}{\rho} \sum_{i \in N} v_i(x_i)^\rho\Big) = \mfrac{\partial v_i(x_i)}{\partial x_{ij}} v_i(x_i)^{\rho-1}$.
By Euler's Theorem for homogeneous functions (Theorem~\ref{thm:euler}), we have
\begin{align}
\frac{\partial v_i(x_i)}{\partial x_{ij}} v_i(x_i)^{\rho-1}  =  \frac{\partial v_i(x_i)}{x_{ij}} \Big(r^{-1} \sum_{\ell \in M} x_{i\ell} \frac{\partial v_i(x_i)}{x_{i\ell}}\Big)^{\rho-1} \label{eq:euler}
\end{align}
Similarly,
\begin{align}
\frac{1}{\rho} \sum_{i \in N} v_i(x_i)^\rho = \frac{1}{\rho} \sum_{i \in N} \Big(r^{-1} \sum_{j \in M} x_{ij} \frac{\partial v_i(x_i)}{\partial x_{ij}}\Big)^\rho\label{eq:euler2}
\end{align}

Therefore for any allocation $\x$, we can compute both the objective function value and the gradient of the objective function using only the gradients of $v_1,\dots,v_n$. The final ingredient we need is an initial ball guaranteed to contain the optimum; we can simply enclose the entire feasible region in a ball of constant radius.

Thus we get the following iterative algorithm for computing the equilibrium pricing rule:
\begin{enumerate}
\item Run the ellipsoid method (or any other suitable convex optimization algorithm) to solve Program~\ref{pro:ces}.
\item At the start of each iteration, ask each agent $i$ for the gradient of $v_i$ at the current point $\x$.
\item Whenever the algorithm requires the gradient of the objective function at $\x$, compute it via Equation~\ref{eq:euler}.
\item Whenever the algorithm requires the objective function value at $\x$, compute it via Equation~\ref{eq:euler2}.
\end{enumerate}

Lemma~\ref{lem:converge} immediately implies correctness and polynomial-time convergence. 


\subsection{Eliciting the gradients of valuations}

The above algorithm (as well as Lemma~\ref{lem:x-to-q}) requires us to have access the gradients of agents' valuations. We could simply ask each agent for this information explicitly; depending on the application domain, this may or may not be reasonable. An alternative approach is to relate $\nabla v_i(x_i)$ to agent $i$'s willingness to pay. For example, consider the following query to agent $i$: ``Suppose you have already bought the bundle $x_i$. What is the smallest marginal price for good $j$ such that you would not buy more of good $j$?" The KKT conditions for agent $i$'s demand set imply that the answer to this question is exactly $\mfrac{\partial v_i(x_i)}{\partial x_{ij}}$ (assuming that the agent would not buy more if she is indifferent). 

Such a query could be implemented in a variety of ways. One possibility would be gradually increasing the hypothetical marginal price of good $j$ in a continuous fashion, and asking agent $i$ to say ``stop" when she would no longer buy more of good $j$ (in a ``moving-knife"-like fashion). Also, rather asking agents about hypothetical marginal prices, one could build the necessary marginal prices into an actual pricing rule, e.g., even one as simple as $p(x_i) = \sum_{j \in M} c_j x_{ij}$. The choice of implementation would depend heavily on the specific problem setting; our point here is that there are a variety of ways to elicit $\nabla v_i(x_i)$ via queries about what agent $i$ would purchase in different (hypothetical) situations.

\section{Truthfulness}\label{sec:truthful}

An alternative approach to implementation is via \emph{truthful} mechanisms. Walrasian equilibria are generally not truthful: agents can sometimes create more favorable equilibrium prices by lying about their preferences. In this section, we present a truthful mechanism for optimizing CES welfare in the case of a single good (Theorem~\ref{thm:truthful}), and show that it is unique up to additive constants in the payment rule (Theorem~\ref{thm:truthful-unique}). Note that uniqueness beyond additive constants in the payment rule can never be achieved without additional assumptions (e.g., individual rationality), since such constants do not affect the behavior of agents. 

Before formally stating and proving these results, we mention an important distinction between this section and Section~\ref{sec:comp}. Section~\ref{sec:comp} is an implementation of the WE from Theorem~\ref{thm:main} (which we know maximizes CES welfare). In contrast, the truthful mechanism from this section is an implementation of CES welfare maximization directly, not an implementation of the WE from Theorem~\ref{thm:main}. Indeed, we know that the payment rule from Theorem~\ref{thm:main} is not truthful, so we must consider a different payment rule if we desire truthfulness.

To define our truthful mechanism we need the following two lemmas, whose proofs appear in Appendix~\ref{sec:proofs}. The first states that for a single good, homogeneous and differentiable functions take a very simple form. The second states that for a single good, the maximum CES welfare allocations take a very simple form; also, for $\rho \ne 1$, the optimum is unique.

\begin{restatable}{lemma}{lemHomoOneGood}\label{lem:homo-m=1}
Let $f: \bbrpos \to \bbrpos$ be differentiable and homogeneous of degree $r$. Then there exists $c \in \bbrpos$ such that $f(x) = c x^r$.
\end{restatable}

\begin{restatable}{lemma}{lemOptOneGood}\label{lem:m=1-x}
Let $m=1$ and $v_i(x_i) = w_i x_i^r$ for all $i \in N$ where $r \in (0,1]$. Then $\rho \in (0,1]$ and $r\rho \ne 1$, $\x \in \Psi(\rho)$ if and only if
\[
x_i = \frac{{w_i}^{\frac{\rho}{1-r\rho}}}{\sum_{k \in N} {w_k}^{\frac{\rho}{1-r\rho}}}
\]
If $\x \in \Psi(\rho)$ and $r = \rho = 1$, then whenever $x_i > 0$, $w_i = \max_{k \in N} w_k$.
\end{restatable}

We now define our mechanism. For $\rho = 1$, the VCG mechanism truthfully maximizes utilitarian welfare~\cite{Nisan2007}, so assume $\rho \in (0,1)$. We ask each agent $i$ to report $w_i$ (where $v_i(x_i) = w_i \cdot x_i^r$), assume the $w_i$'s are truthful, and output the (unique) optimal allocation $\x \in \Psi(\rho)$ according to Lemma~\ref{lem:m=1-x}. Let $\B = b_1,\dots,b_n$ be the vector of reported $w_i$'s. We then charge each agent $i$ the following payment:\footnote{Although this integral does not have a simple closed form, it can be expressed via the hypergeometric function.}
\begin{align}
p_i(\B) =  \frac{r\rho}{1-r\rho} \Big(\sum_{k \ne i} b_k^\frac{\rho}{1-r\rho}\Big) \int_{b = 0}^{b_i}  \frac{b^\frac{r\rho}{1-r\rho}}{\big(b^\frac{\rho}{1-r\rho} + \sum_{k \ne i} b_k^\frac{\rho}{1-r\rho}\big)^{r+1}}\dif b \label{eq:payment}
\end{align}
This payment is chosen so that the derivative of agent $i$'s utility at $b_i = w_i$ is 0. In particular, let $x_i(\B)$ denote agent $i$'s bundle under reports $\B$. Then we will have $\mfrac{\partial v_i(x_i(\B))}{\partial b_i} = rw_i\mfrac{\partial x_i(\B)}{\partial b_i} x_i(\B)^{r-1}$, and $\mfrac{\partial p_i(\B)}{\partial b_i} = rb_i\mfrac{\partial x_i(\B)}{\partial b_i} x_i(\B)^{r-1}$,
so the derivative of agent $i$'s overall utility will be $\mfrac{\partial u_i(\B)}{\partial b_i} = r (w_i - b_i) \mfrac{\partial x_i(\B)}{\partial b_i} x_i(\B)^{r-1}$. This will imply that it is optimal for agent $i$ to truthfully report $b_i = w_i$.

\begin{restatable}{theorem}{thmTruthful}\label{thm:truthful}
Assume $m=1$, and that each $v_i$ is homogenous of degree $r$ (with $r$ publicly known), concave, and differentiable. Then for all $\rho \in (0,1)$, there is a truthful mechanism which outputs an allocation $\x \in \Psi(\rho)$.
\end{restatable}

\begin{proof}
Since VCG satisfies the claim for $\rho = 1$, assume $\rho \in (0,1)$. Let $\x(\B)$ denote the allocation outputted given reports $\B$, and let $x_i(\B)$ denote agent $i$'s bundle: formally, $x_i(\B) = \mfrac{{b_i}^{\frac{\rho}{1-r\rho}}}{\sum_{k \in N} {b_k}^{\frac{\rho}{1-r\rho}}}$. Since $m=1$, Lemma~\ref{lem:homo-m=1} implies that for all $i \in N$, there exists $w_i \in \bbrpos$ such that $v_i(x) = w_i\cdot x^r$ for all $x \in \bbrpos$. Then by Lemma~\ref{lem:m=1-x}, $x_i(\B) \in \Psi(\rho)$, so it remains to prove truthfulness.

Since we assume that each agent's valuation is not identically zero, we have $w_i > 0$. Also, by concavity and monotonicity of $v_i$, we have $r \in (0,1]$. Thus $0 < r\rho < 1$. Since we also have $b_i > 0$, all denominators in $p_i(\B)$ are nonzero and thus $p_i(\B)$ is well-defined.

 Let $v_i(\B) = v_i(x_i(\B)) = w_i x_i(\B)^r$ for brevity, and let $u_i(\B) = v_i(\B) - p_i(\B)$ denote agent $i$'s resulting utility under bids $\B$. Note that $x_i(\B)$, $v_i(\B)$, $p_i(\B)$, and $u_i(\B)$ are all differentiable with respect to $b_i$. Also let $\alpha = \mfrac{\rho}{1-r\rho}$ for brevity; then $p_i(\B) = r\alpha (\sum_{k \ne i} b_k^\alpha) \int_{b = 0}^{b_i} \mfrac{b^{r\alpha}}{(b^\alpha + \sum_{k \ne i} b_k^\alpha)^{r+1}}\dif b$ and $x_i(\B) = \mfrac{b_i^\alpha}{\sum_{k \in N} b_k^\alpha}$.

To prove truthfulness, we need to show that $w_i \in \argmax_{b_i \in \bbr_{> 0}} u_i(\B)$, i.e., truthfully reporting $w_i$ is an optimal strategy for agent $i$.\footnote{Note that $u_i(\B)$ is not concave in $b_i$, since $p_i(\B)$ is not convex in $b_i$. Thus the KKT conditions do not apply, so we will have to use a different approach.}  Since $u_i(\B)$ is differentiable with respect to $b_i$, we have $\mfrac{\partial u_i(\B)}{\partial b_i} = \mfrac{\partial v_i(\B)}{\partial b_i} - \mfrac{\partial p_i(\B)}{\partial b_i}$. The first term on the right hand side is
\[
\frac{\partial v_i(\B)}{\partial b_i} = rw_i\frac{\partial x_i(\B)}{\partial b_i} x_i(\B)^{r-1}
\]
The second term is
\begin{align*}
\frac{\partial p_i(\B)}{\partial b_i} =&\ r \alpha \Big(\sum_{k \ne i} b_k^\alpha\Big) \frac{b_i^{r\alpha}}{(b_i^\alpha + \sum_{k \ne i} b_k^\alpha)^{r+1}}\\
=&\ r \alpha \Big(\sum_{k \ne i} b_k^\alpha\Big) \frac{b_i^{r\alpha}}{(\sum_{k \in N} b_k^\alpha)^{r+1}}\\
=&\ r \alpha \Big(\sum_{k \ne i} b_k^\alpha\Big) \frac{b_i^\alpha}{(\sum_{k \in N} b_k^\alpha)^2} \Big(\frac{b_i^\alpha}{\sum_{k \in N} b_k^\alpha}\Big)^{r-1}\\
=&\ r \alpha \Big(\sum_{k \ne i} b_k^\alpha\Big) \frac{b_i^\alpha}{(\sum_{k \in N} b_k^\alpha)^2} x_i(\B)^{r-1}
\end{align*}
Conveniently, we have $\mfrac{\partial}{\partial b_i} \big(\mfrac{b_i^\alpha}{\sum_{k \in N} b_k^\alpha}\big) =  \alpha \big(\sum_{k \ne i} b_k^\alpha\big) \mfrac{b_i^{\alpha-1}}{(\sum_{k \in N} b_k^\alpha)^2}$. Thus
\begin{align*}
\frac{\partial p_i(\B)}{\partial b_i} =&\ r b_i \frac{\partial}{\partial b_i} \Big(\frac{b_i^\alpha}{\sum_{k \in N} b_k^\alpha}\Big) x_i(\B)^{r-1}\\
=&\ r b_i \frac{\partial x_i(\B)}{\partial b_i} x_i(\B)^{r-1}
\end{align*}
Therefore
\[
\frac{\partial u_i(\B)}{\partial b_i} = r (w_i - b_i) \frac{\partial x_i(\B)}{\partial b_i} x_i(\B)^{r-1}
\]
Since $\frac{\partial x_i(\B)}{\partial b_i} > 0$ and $x_i(\B)^{r-1} > 0$ for all $b_i$, this implies
\begin{enumerate}
\item For all $b_i < w_i$, $\mfrac{\partial u_i(\B)}{\partial b_i} > 0$.
\item For all $b_i > w_i$, $\mfrac{\partial u_i(\B)}{\partial b_i} > 0$.
\item For $b_i = w_i$, $\mfrac{\partial u_i(\B)}{\partial b_i} = 0$.
\end{enumerate}

Therefore $w_i \in \argmax_{b_i \in \bbr_{> 0}}$ (in fact, $w_i$ is the unique maximizer). We conclude that the mechanism is truthful.
\end{proof}

From a technical standpoint, the harder task is proving that this mechanism is unique (up to additive constants in the payment rule). We assume without loss of generality that the mechanism asks each agent $i$ to report $w_i$, and let $\B=b_1,\dots,b_n$ be the vector of reported $w_i$'s. We use the standard notation of $(\B_{-i}, b_i')$ to denote the vector where the $i$th entry is $b_i'$, and the $k$th entry for each $k\ne i$ is $b_k$.


The proof takes a real analysis approach, with Kirszbraun's Theorem for Lipschitz extensions~\cite{kirszbraun1934} playing a central role. On a high level, the proof proceeds as follows: (1) we establish some basic properties of the payment rule, (2) we show that the payment rule must be Lipschitz continuous not including $b_i = 0$, (3) there exists a Lipschitz extension $\phat$ including $b_i = 0$ (Kirszbraun's Theorem), (4) since $\phat$ is Lipschitz, it is differentiable almost everywhere and is equal to the integral of its derivative, (5) since it has the same derivative (when defined) as the payment rule from Theorem~\ref{thm:truthful}, the payment rules are equal (up to the constant of integration).

\begin{theorem}\label{thm:truthful-unique}
Assume $m=1$, and that each $v_i$ is homogenous of degree $r$ (with $r$ publicly known), concave, and differentiable. Fix $\rho \in (0,1)$, and let $\Gamma$ be a truthful mechanism which outputs an allocation $\x \in \Psi(\rho)$. Then the allocation rule is the same as in Theorem~\ref{thm:truthful}, and the payment rule $p_i(\B)$ is the same up to an additive constant.
\end{theorem}

\begin{proof}
\textbf{Part 1: Setup and basic properties.}
Since there is a unique optimal allocation (Lemma~\ref{lem:m=1-x}), $\Gamma$ must take $\B = (b_1,\dots, b_n)$ as honest and output the same allocation $\x(\B)$. 

It remains to consider the payment rule. Let $p_i(\B)$ denote the payment rule from Theorem~\ref{thm:main}, and let $\tilp(\B)$ denote the payment rule for $\Gamma$. Given reports $\B$, define $v_i(\B)$ as before, and let $u_i(\B) = v_i(\B) - \tilp(\B)$ be agent $i$'s resulting utility under $\Gamma$. From the point of view of a given agent $i$, the other agents' reports $\B_{-i}$ can be treated as a constant. Thus for brevity, write $x_i(b) = x_i(\B_{-i}, b)$, $p_i(b) = p_i(\B_{-i}, b)$, and $\tilp(b) = \tilp(\B_{-i}, b)$ for each $i \in N$.

Fix an $i \in N$. Since $\Gamma$ is truthful, we must have $w_i \in \argmax_{b_i \in \bbr_{>0}} u_i(\B)$. Then by definition of $u_i(\B)$, we have $w_i \in \argmax_{b_i \in \bbr_{>0}} \big(w_i x_i(b_i)^r - \tilp(b_i)\big)$. Since $w_i$ could be any element of $\bbrspos$, and $\Gamma$ must be agnostic to $w_i$, we must have $b \in \argmax_{b_i \in \bbr_{>0}} \big(b x_i(b_i)^r - \tilp(b_i)\big)$ for all $b \in \bbrspos$.

We first claim that $\tilp(b)$ is nondecreasing. Suppose the opposite: then there exists exists $b > b'$ such that $\tilp(b) < \tilp(b')$. But this means that if $w_i = b'$, reporting $b_i = w_i$ is never an optimal strategy, because the payment can be decreased by reporting $b_i = b$, and $x_i(b) \ge x_i(b')$ (since $x_i(b)$ is nondecreasing). Thus $\tilp(b)$ is nondecreasing.

\medskip

\textbf{Part 2: $\tilp$ is Lipschitz continuous.} Fix an arbitrary $b_i > 0$. Since $b_k > 0$ for all $k \ne i$, it can be seen from the definition of $x_i(b)$ that $x_i(b)^r$ is continuously differentiable on $[0,b_i]$. Therefore the maximum of $\mfrac{\dif x_i(b)^r}{\dif b}$ is a Lipschitz constant for $x_i(b)^r$, so $x_i(b)^r$ is Lipschitz continuous on $[0,b_i]$. Let $\kappa$ be this Lipschitz constant: then for all $b, b' \in [0, b_i]$, $|x_i(b)^r - x_i(b')^r| \le \kappa |b - b'|$.

We claim that $\tilp$ is Lipschitz continuous on $(0, b_i]$ with constant $b_i \kappa$. Suppose the opposite: then there exist $b, b' \in (0, b_i]$ such that $|\tilp(b) - \tilp(b')| > b_i \kappa |b - b'|$. Assume without loss of generality that $b > b'$. Since $\tilp$ and $x_i$ are both nondecreasing, we then have $\tilp(b) - \tilp(b') > b_i \kappa (b - b')$ and $x_i(b)^r - x_i(b')^r \le \kappa(b-b')$.

Since $b \in \argmax_{b_i \in \bbr_{>0}} \big(b x_i(b_i)^r - \tilp(b_i)\big)$, we have $b x_i(b)^r - \tilp(b) \ge b x_i(b')^r - \tilp(b')$ and thus $b(x_i(b)^r - x_i(b')^r) \ge \tilp(b) - \tilp(b')$. Therefore
\[
b\kappa(b-b') \ge b\big(x_i(b)^r - x_i(b')^r\big) \ge \tilp(b) - \tilp(b') > b_i \kappa (b - b')
\]
Therefore $\mfrac{b\kappa}{b_i \kappa} > 1$, which contradicts $b \le b_i$. Therefore $\tilp$ is Lipschitz continuous on $(0, b_i]$. 

\medskip

\textbf{Part 3: Kirszbraun's Theorem.} Thus by Kirszbraun's Theorem~\cite{kirszbraun1934}, $\tilp$ has a Lipschitz extension to $[0,b_i]$: that is, there exists $\phat: [0,b_i] \to \bbrpos$ such that $\phat$ is Lipschitz continuous on $[0,b_i]$, and $\phat(b) = \tilp(b)$ for $b \in (0, b_i]$.

\medskip

\textbf{Part 4: $\phat$ is the integral of its derivative.}
Lipschitz continuity implies absolute continuity~\cite{Royden1988}, so $\phat$ is absolutely continuous on $[0, b_i]$. Thus by the Fundamental Theorem of Lebesgue Calculus~\cite{Royden1988}, $\phat$ is differentiable almost everywhere on $[0, b_i]$, its derivative $\mfrac{\dif \phat(b)}{\dif b}$ is integrable over $[0, b_i]$, and
\[
\phat(b_i) - \phat(0) = \int_{b = 0}^{b_i} \frac{\dif \phat(b)}{\dif b} \dif b
\]

\textbf{Part 5: The derivatives of $\phat$ and $p_i$ match, so $\phat = p_i + c$.}  Consider a $b > 0$ at which $\phat$ is differentiable. Then $\tilp$ is also differentiable, so $b \in \argmax_{b_i \in \bbr_{>0}} \big(b x_i(b_i)^r - \tilp(b_i)\big)$ implies $\mfrac{\dif \tilp(b)}{\dif b} = b\mfrac{\dif}{\dif b} (x_i(b)^r) = rb\mfrac{\dif x_i(b)}{\dif b} x_i(b)^{r-1}$.\footnote{Note that the $b$ in $bx_i(b_i)^r$ is a constant from the point of view of the argmax, so it is treated as a constant by the derivative. To be technically precise, we have $(\frac{\dif}{\dif b_i} b x_i(b_i)^r)|_{b_i = b} = rb\frac{\dif x_i(b)}{\dif b} x_i(b)^{r-1}$.} We showed in the proof of Theorem~\ref{thm:truthful} that $\mfrac{\partial p_i(\B)}{\partial b_i} = r b_i \frac{\partial x_i(\B)}{\partial b_i} x_i(\B)^{r-1}$; equivalently, $\mfrac{\dif p_i(b)}{\dif b} = r b \frac{\dif x_i(b)}{\dif b} x_i(b)^{r-1}$. Therefore for all $b > 0$ at which $\phat$ is differentiable, we have $\mfrac{\dif \phat(b)}{\dif b} = \mfrac{\dif p_i(b)}{\dif b}$.

Since $\phat$ is differentiable almost everywhere, we have $\mfrac{\dif \phat(b)}{\dif b} = \mfrac{\dif p_i(b)}{\dif b}$ almost everywhere. Thus $ \mfrac{\dif p_i(b)}{\dif b}$ is also integrable over $[0,b_i]$, and $\int_{b = 0}^{b_i} \mfrac{\dif p_i(b)}{\dif b} \dif b = \int_{b = 0}^{b_i} \mfrac{\dif \phat(b)}{\dif b} \dif b$~\cite{Royden1988}. Therefore 
\begin{align*}
\phat(b_i) =&\ \phat(0) + \int_{b = 0}^{b_i} \frac{\dif p_i(b)}{\dif b} \dif b\\
=&\ \phat(0) + p_i(b_i)
\end{align*}
where the second equality is from the definition of $p_i(\B)$.

Therefore for all $b_i > 0$, $\tilp(b_i) = \phat(0) + p_i(b_i)$, and so $\tilp(\B) = \phat(0) + p_i(\B)$ for all $\B$. Since this holds for all $i \in N$, $\tilp(\B)$ is exactly the payment rule from Theorem~\ref{thm:truthful}, up to the additive constant of $\phat(0)$.
\end{proof}

It is worth noting that this truthful payment rule is quite complex; in particular, it may be hard to convince agents that it is actually in their best interest to be truthful. In contrast, the Walrasian pricing rule from Theorem~\ref{thm:main} is much simpler and more intuitive. That pricing rule is not truthful, but perhaps formal truthfulness is not crucial if a practical iterative implementation is possible. We do not claim that our algorithm from Section~\ref{sec:comp} is truly practical, but it could be a step in the right direction.

\section{Sybil attacks}\label{sec:sybil}

In Sections~\ref{sec:comp} and \ref{sec:truthful}, we discussed two alternative approaches to implementation: an iterative query-based algorithm, and a truthful mechanism. However, there is an additional crucial issue which any practical implementation must address: since our pricing rule $p(x_i) = (\sum_{j \in M} q_j x_{ij})^{1/\rho}$ is strictly convex for $\rho < 1$, agents have an incentive to create fake identities. In particular, an agent can decrease her payment while receiving the same bundle by splitting the payment over multiple fake identities.\footnote{Note that for linear prices there is no such incentive.} This is known as a \emph{Sybil attack}. The truthful payment rule from Section~\ref{sec:truthful} is not strictly convex everywhere, but it is strictly convex on some intervals, and thus has the same vulnerability. 

\paragraph{Model of Sybil attacks.} We model this as follows. Let $\kappa$ denote the cost of creating a new identity. The cost could reflect inconvenience, risk of getting caught, or other factors, and would depend on the nature of the system. Let $\eta_i$ be the \emph{multiplicity} of agent $i$, i.e., the number of identities agent $i$ controls in the system. This includes both fake identities and the agent's single real identity, so we assume that $\eta_i \in \bbnspos$. For convex $p$, multiplicity $\eta_i$, and a desired bundle for purchase, it is always optimal for agent $i$ to split the purchase evenly across her identities.\footnote{This is essentially a multidimensional version of Jensen's inequality; see, e.g.,~\cite{Neuman1990}.} Thus we can assume that each identity purchases the same bundle $x_i$, and we define agent $i$'s utility as
\[
u_i(x_i, \eta_i) = v_i(\eta_i x_i) - \eta_i p(x_i) - \eta_i \kappa
\]
We do not claim that this perfectly models the reality of Sybil attacks; for example, the identity creation cost is arguably sublinear (one someone has created a single fake identity, creating more might become easier). Our goal here is simply to show formally that at least in some cases, CES welfare maximization cannot be robust to Sybil attacks in general.

\paragraph{Walrasian equilibrium.} We focus on a Walrasian model of Sybil attacks; the analogous analysis for truthful mechanisms is left as an open question. We define each agent's \emph{Sybil demand set} by
\[
D_i(p) = \argmax_{x_i \in \bbrpos^m, \eta_i \in \bbnspos}\ u_i(x_i, \eta_i)
\]
Note that we require $\eta_i \in \bbnspos$. We define a \emph{Sybil Walrasian equilibrium} (SWE) to be an allocation $\x$, payment rule $p$, and vector of multiplicities $\bfeta = \eta_1,\dots,\eta_n$ such that
\begin{enumerate}
\item Each agent receive a bundle in her demand set: $(x_i, \eta_i) \in D_i(p)$ for all $i \in N$.
\item The market clears: for all $j \in M$, $\sum_{i \in N} x_{ij} \le 1$. Furthermore, for any $j\in M$ with nonzero cost\footnote{Recall that good $j$ has ``nonzero cost" in our pricing rule if $q_j > 0$.}, $\sum_{i \in N} x_{ij} = 1$.
\end{enumerate}

In this section, we will focus on the case of homogeneity degree $r=1$. The following lemma states that for any pricing rule, a rational agent either creates no fake identities (i.e, $\eta_i = 1$), or creates an unbounded number (and consequently the demand set is empty).

\begin{lemma}\label{lem:sybil-demand}
Assume each $v_i$ is concave, differentiable, and homogeneous of degree 1. Let $\rho \in (0,1]$, define $p$ as in Theorem~\ref{thm:main}, and let $\x \in \Psi(\rho)$. Then we have
\[
D_i(p) = \begin{cases}
(x_i, 1) &\ \text{if } v_i(x_i) (1-\rho) \le \kappa\\
\emptyset &\ \text{otherwise}
\end{cases}
\]
where $x_i$ is agent $i$'s bundle in $\x$.
\end{lemma}

\begin{proof}
When $v_i$ is homogeneous of degree 1, for any bundle $y_i$, we have $u_i(y_i, \eta_i) = \eta_i v_i(y_i) - \eta_i p(y_i) - \eta_i \kappa = \eta_i\big(v_i(y_i) - p(y_i) - \kappa\big)$. Thus given a choice of $\eta_i$, $y_i$ must be chosen to maximize $v_i(y_i) - p(y_i) - \kappa$. Let $\x \in \Psi(\rho)$: then by Theorem~\ref{thm:main} $y_i$ optimizes $v_i(y_i) - p(y_i)$ (and thus $v_i(y_i) - p(y_i) - \kappa$) if and only if $y_i = x_i$. Therefore the demand set is equal to
\[
D_i(p) = \Big(x_i,\  \argmax_{\eta_i \in \bbnspos}\ \eta_i \big(v_i(x_i) - p(x_i) - \kappa\big) \Big)
\]
That is, the demanded bundle must always be $x_i$, and $\eta_i$ is optimized accordingly. 

By Corollary~\ref{cor:utility-price}, $p(x_i) = \rho v_i(x_i)$, so $v_i(x_i) - p(x_i) - \kappa = v_i(x_i)(1-\rho) - \kappa$. Thus if $v_i(x_i) (1-\rho) \le \kappa$, then $1$ is an optimal choice for $\eta_i$, so $(x_i, 1) \in D_i(p)$. However, if $v_i(x_i) (1-\rho) > \kappa$, there is no optimal choice for $\eta_i$: specifically, $\eta_i$ goes to infinity. Thus if $v_i(x_i) (1-\rho) > \kappa$, $D_i(p) = \emptyset$.
\end{proof}

This immediately implies that if $\x \in \Psi(\rho)$ satisfies $v_i(x_i) (1-\rho) \le \kappa$ for all $i \in N$, the convex pricing rule from Theorem~\ref{thm:main} is naturally robust to Sybil attacks.

\begin{theorem}\label{thm:sybil-r=1-good}
Assume each $v_i$ is concave, differentiable, and homogeneous of degree 1. Let $\x \in \Psi(\rho)$ for $\rho \in (0,1]$, and define $p$ as in Theorem~\ref{thm:main}. Then if $v_i(x_i) (1-\rho) \le \kappa$ for all $i \in N$, $(\x, p, \one)$ is a SWE.
\end{theorem}

\begin{proof}
By Lemma~\ref{lem:sybil-demand}, we have $(x_i, 1) \in D_i(p)$ for all $i \in N$ in this case. Theorem~\ref{thm:main} implies that the market clearing condition is met, so $(\x, p, \one)$ is a SWE.
\end{proof}

In other words, if the identity creation cost is small, $\rho$ is close to 1, and/or agents valuations are not too large, we need not worry about Sybil attacks. As discussed in Section~\ref{sec:results}, this suggests one possible way for a social planner to choose a value of $\rho$: estimate $\kappa$ and the magnitude of valuations, and choose $\rho$ to be as small as possible without incentivizing Sybil attacks.

On the other hand, if $v_i(x_i) (1-\rho) > \kappa$, how bad are the consequences? Theorem~\ref{thm:sybil-r=1-bad} states that an agent's valuation at equilibrium has a hard cap at $\mfrac{\kappa}{1-\rho}$. This provides a hard maximum on the CES welfare in any SWE with $p$ thus defined: in particular, the CES welfare is at most $\big(\sum_{i \in N} \big(\frac{\kappa}{1-\rho}\big)^\rho\big)^{1/\rho} = n^{1/\rho} \frac{\kappa}{1-\rho}$. In general, each $v_i(x_i)$ (and thus the CES welfare) can be arbitrarily large, so Theorem~\ref{thm:sybil-r=1-bad} implies an unbounded ratio between the optimal CES welfare and that of any SWE with this $p$.

\begin{theorem}\label{thm:sybil-r=1-bad}
Assume each $v_i$ is concave, differentiable, and homogeneous of degree 1. Let $\rho \in (0,1]$, and define $p$ as in Theorem~\ref{thm:main}. Then for any allocation $\x$ and multiplicities $\bfeta$ such that $(\x, p, \bfeta)$ is a SWE, we have
\[
v_i(x_i) \le \frac{\kappa}{1-\rho}
\]
\end{theorem}

\begin{proof}
Suppose $(\x, p, \bfeta)$ is a SWE for some allocation $\x$ and multiplicities $\bfeta$: then each $(x_i, \eta_i) \in D_i(p)$ for all $i \in N$; Thus $D_i(\p) \ne \emptyset$, so Lemma~\ref{lem:sybil-demand} implies that $v_i(x_i) (1-\rho) \le \kappa$, and consequently, $v_i(x_i) \le \frac{\kappa}{1-\rho}$.
\end{proof}

The next natural question is, can we circumvent this by using a different pricing rule? Theorem~\ref{thm:sybil-other-p} answers this in the negative. The counterexample uses an instance with a single good; recall that $x_i$ denotes a scalar in this case.

\begin{theorem}\label{thm:sybil-other-p}
Let $m=1$, $v_1(x_1) = w x_1$, and $v_i(x_i) = x_i$ for all $i \ne 1$. Let $(\x, p, \bfeta)$ be any SWE. Then for all $i \ne 1$,
\[
v_i(x_i) \le \frac{\kappa}{w-1}
\]
\end{theorem}

\begin{proof}
Let $(\x, p, \bfeta)$ be any SWE. Fix an arbitrary $i \in N$. As in Lemma~\ref{lem:sybil-demand}, we have $u_i(x_i, \eta_i) = \eta_i (v_i(x_i) - p(x_i) - \kappa)$. Since $(x_i, \eta_i) \in D_i(p)$, we must have $\eta_i \in \argmax_{\eta_i' \in \bbnspos} \eta_i (v_i(x_i) - p(x_i) - \kappa)$ (note that we are not assuming anything about the bundle $x_i$). Since $\argmax_{\eta_i' \in \bbnspos} \eta_i' (v_i(x_i) - p(x_i) - \kappa)$ cannot be the empty set, we must have $v_i(x_i) \le p(x_i) + \kappa$ and $\eta_i = 1$.

Focusing on agent 1, we further claim that $v_1(x_i) \le p(x_i) + \kappa$ for any $i \ne 1$. Suppose not: then agent 1 could purchase $x_i$ and set $\eta_1 = \infty$ to increase her utility. Thus $v_1(x_i) \le p(x_i) + \kappa$ for each $i \ne 1$. Now looking at the optimization for $i \ne 1$, we have $v_i(x_i) \ge p(x_i)$. Combining this with $v_1(x_i) \le p(x_i) + \kappa$, we get $v_1(x_i) \le v_i(x_i) + \kappa$.

Plugging in our definitions of $v_1$ and $v_{i \ne 1}$, we get $w x_i \le x_i + \kappa$, so $x_i (w - 1) \le \kappa$. Substituting back in the definition of $v_i$, we get $v_i(x_i) \le \mfrac{\kappa}{w-1}$ for all $i \ne 1$, as required.
\end{proof}

Although the bound in Theorem~\ref{thm:sybil-other-p} is different from that in Theorem~\ref{thm:sybil-r=1-bad}, the implication is the same: this is a hard maximum on the value obtained by any agent other than agent 1. As $\kappa$ goes to zero, the fraction of the good agent 1 receives approaches 1, so the outcome approaches the maximum utilitarian welfare outcome (where agent 1 receives the entirety of the good). Therefore by Theorem~\ref{thm:rho=1-bad}, the CES welfare at any Sybil Walrasian equilibrium (for any pricing rule) can be arbitrarily bad in comparison to the optimal CES welfare. Thus in general, when Sybil attacks are possible, it is impossible to implement any bounded approximation of CES welfare maximization in Walrasian equilibrium.

\section{Negative results}\label{sec:counter}

Even when Sybil attacks are not possible, there are limitations to implementation in WE. This section presents several relevant counterexamples.

\subsection{Linear pricing poorly approximates CES welfare for $\rho \ne 1$}

Recall that for an allocation $\x$, $\Phi(\rho, \x)$ denotes the CES welfare of $\x$. In contrast, $\Psi(\rho)$ denotes the set of allocations with optimal CES welfare with respect to $\rho$.

Our first negative result relates to linear pricing. In particular, can linear pricing guarantee a reasonable approximation of CES welfare? We show that the answer is no, justifying the need for nonlinear pricing. In particular, for any $\rho \in (0,1)$, the gap between the CES welfare of any linear pricing equilibrium and the optimal CES welfare can be arbitrarily large.

Note that as $\rho$ goes to zero, $\frac{1}{\rho} - 1$ goes to infinity, so the denominator of the bound (and thus the gap in CES welfare) in the following theorem can indeed be arbitrarily large.

\begin{theorem}\label{thm:rho=1-bad}
Let $m=1$, $\rho \in (0,1]$, $v_1(x) = (1+\ep) x$ for some $\ep > 0$, and $v_i(x) = x$ for all $i \ne 1$. Suppose $(\x, p)$ is a WE where $p$ is linear. Then
\[
\frac{\Phi(\rho, \x)}{\max_{\y} \Phi(\rho, \y)} \le \frac{1+\ep}{n^{\frac{1}{\rho}-1}}
\]
\end{theorem}

\begin{proof}
By the First Welfare Theorem, $\x$ must maximize utilitarian (i.e., $\rho = 1$) welfare. Thus by Lemma~\ref{lem:m=1-x}, $\x$ must give the entire good to agent 1: $x_1 = 1$ and $x_i = 0$ for $i \ne 1$. Thus the CES welfare of $\x$ with respect to $\rho$ is
\[
\Phi(\rho, \x) = \Big(\sum_{i \in N} v_i(x_i)^\rho\Big)^{1/\rho} = \big( (1+\ep)^\rho \big)^{1/\rho}
\]
In contrast, consider the allocation $\y$ such that $y_i = 1/n$ for all $i \in N$:
\[
\Phi(\rho, \y) = \Big( \sum_{i \in N} v_i(1/n)^{\rho} \Big)^{1/\rho} \ge \Big(\sum_{i \in N} (1/n)^\rho \Big)^{1/\rho} = \Big(n (1/n)^{1/\rho} \Big)^{1/\rho} = n^{\frac{1}{\rho} - 1}
\]
Thus $\max_{\y} \Phi(\rho, \y) \ge n^{\frac{1}{\rho} - 1}$, as required.
\end{proof}

\subsection{Theorem~\ref{thm:main} does not extend to nonuniform homogeneity degrees}

In this section, we show that for all $\rho \in (0,1)$, Theorem~\ref{thm:main} does not extend to the case where different $v_i$'s have different homogeneity degrees. This shows that our result is tight in the sense that it is necessary to require the same homogeneity degree.

We begin with the following lemma, which is a standard property of strictly concave and differentiable functions: it essentially states that any such function is bounded above by any tangent line. This lemma is sometimes called the ``Rooftop Theorem".

\begin{lemma}\label{lem:rooftop}
Let $f: \bbr \to \bbr$ be strictly concave and differentiable. Then for all $a,b \in \bbr$ where $a\ne b$, $f(a) < f(b) + f'(b)(a-b)$, where $f'$ denotes the derivative of $f$.
\end{lemma}

The next lemma is also quite standard; we provide a proof for completeness.

\begin{lemma}\label{lem:concave-bound}
Let $f: \bbr \to \bbr$ be strictly concave and differentiable, and let $x, a_1, \dots, a_k$ be nonnegative reals such that $\sum_{i=1}^k a_i = 0$. Then $\sum_{i =1}^k f(x + a_i) < kf(x)$.
\end{lemma}

\begin{proof}
The lemma follows from Lemma~\ref{lem:rooftop} and arithmetic:
\begin{align*}
\sum_{i =1}^k f(x + a_i) <&\ \sum_{i=1}^k (f(x) + f'(x)(x+a_i - x))\\
=&\  \sum_{i=1}^k f(x) + f'(x)  \sum_{i=1}^k a_i\\
=&\  \sum_{i=1}^k f(x) + f'(x) \cdot 0\\
=&\  kf(x)
\end{align*}
\end{proof}

We are now ready to present our counterexample.

\begin{theorem}\label{thm:different-r}
Let $n =2$ and $m=1$, and for $x \in \bbrpos$, let $v_1(x) = x$ and $v_2(x) = \sqrt{2x}$. Then for all $\rho \in (0,1)$, there exists no allocation $\x \in \Psi(\rho)$ and pricing rule $p: \bbrpos \to \bbrpos$ such that $(\x, p)$ is a WE.
\end{theorem}

\begin{proof}
Suppose for sake of contradiction that such $\x, p$ do exist. We first claim that $x_1 > x_2$. Suppose the opposite: then $x_2 \ge 1/2 \ge x_1$. Thus $v_2(x_2) \ge 1 > 1/2 \ge v_1(x_1)$. We also have $\mfrac{\partial v_2(x_2)}{\partial x_2} = \mfrac{1}{\sqrt{2x_2}} \le 1 = \mfrac{\partial v_1(x_1)}{\partial x_1}$. Thus $v_2(x_2) > v_1(x_2)$ and $\mfrac{\partial v_2(x_2)}{\partial x_2} \le \mfrac{\partial v_1(x_1)}{\partial x_1}$. Since $\rho < 1$, $\rho - 1 < 0$, so we have $v_2(x_2)^{\rho - 1}\mfrac{\partial v_2(x_2)}{\partial x_2} < v_1(x_1)^{\rho - 1}\mfrac{\partial v_1(x_1)}{\partial x_2}$. But this contradicts $\x \in \Psi(\rho)$, so we have $x_1 > x_2$ as claimed.\footnote{This immediately implies $\frac{\partial v_2(x_2)}{\partial x_2} < 1 =\frac{\partial v_1(x_1)}{\partial x_1}$, which, in combination with $x_1 > x_2$, rules out convex $p$. However, we still need to rule out non-convex $p$.}

Since $(\x, p)$ is a WE, we must have $x_i \in D_i(p)$ for both agents $i$. Thus for any $x \ne x_i$, $v_i(x_i) - p(x_i) \ge v_i(x) - p(x)$. Therefore
\begin{align*}
v_1(x_1) - p(x_1) \ge&\ v_1(x_2) - p(x_2) \quad \text{and} \quad v_2(x_2) - p(x_2) \ge v_2(x_1) - p(x_1)\\
v_1(x_1) + v_2(x_2) - p(x_1) - p(x_2) \ge&\ v_1(x_2) + v_2(x_1) - p(x_1) - p(x_2)\\
v_1(x_1) + v_2(x_2) \ge&\ v_1(x_2) + v_2(x_1)\\
v_1(x_1) - v_1(x_2) \ge&\ v_2(x_1) - v_2(x_2)
\end{align*}
Since $x_1 > 1/2 > x_2$ and $x_1 + x_2 = 1$, let $x_1 = 1/2 + \ep$ and $x_2 = 1/2 - \ep$. Then we have $v_1(x_1) - v_1(x_2) = 2\ep$. For $v_2(x_1) - v_2(x_2)$, we have
\begin{align*}
v_2(x_1) - v_2(x_2) =&\ \sqrt{1+2\ep} - \sqrt{1-2\ep}\\
=&\ \frac{(\sqrt{1+2\ep} - \sqrt{1-2\ep})(\sqrt{1+2\ep} - \sqrt{1-2\ep})}
{\sqrt{1+2\ep} + \sqrt{1-2\ep}}\\
=&\ \frac{(1+2\ep) - (1-2\ep)}{\sqrt{1+2\ep} + \sqrt{1-2\ep}}\\
=&\ \frac{4\ep}{\sqrt{1+2\ep} + \sqrt{1-2\ep}}
\end{align*}
Applying Lemma~\ref{lem:concave-bound} with $f(x) = \sqrt{x}$, $x=1$, $k=2$, and $(a_1, a_2) = (2\ep, -2\ep)$, we get $\sqrt{1+2\ep} + \sqrt{1-2\ep} < 2$. Thus $v_2(x_1) - v_2(x_2) > 4\ep/2 = 2\ep = v_1(x_1) - v_1(x_2)$. However, this contradicts $v_1(x_1) - v_1(x_2) \ge v_2(x_1) - v_2(x_2)$, as we showed above. We conclude that there is no $\x \in \Psi(\rho)$ and pricing rule $p$ such that $(\x, p)$ is a WE.
\end{proof}

\subsection{CES welfare maximization for $\rho \le 0$}

In this section, we show that there is no pricing rule supporting CES welfare maximization for any $\rho < 0$. For $\rho = 0$ (i.e., Nash welfare), the situation is slightly different. We do show, however, that Nash welfare maximization cannot be supported by a differentiable pricing rule.

\begin{theorem}\label{thm:neg-rho-counter}
Consider the instance with $n=2$, $m=1$, $v_1(x) = x$ and $v_2(x) = 2x$. Then for every $\rho < 0$, there is no pricing rule $p$ and allocation $\x \in \Psi(\rho)$ such that $(\x, p)$ is a WE.
\end{theorem}

\begin{proof}
For any $\rho < 0$ and any $\x \in \Psi(\rho)$, we must have $x_1 > x_2$. Assume $(\x, p)$ is a WE for some pricing rule $p$: then $x_1 \in D_1(p)$, so $v_1(x_1) - p(x_1) \ge v_1(x_2) - p(x_2)$. Thus $p(x_1) \le p(x_2) + v_1(x_1) - v_2(x_2) = p(x_2) + x_2 - x_1$. Therefore
\begin{align*}
v_2(x_1) - p(x_1) \ge&\ 2x_1 - (p(x_2) + x_2 - x_1)\\
=&\ 3x_1 - x_2 - p(x_2)\\
>&\ 2x_1 - p(x_2)\\
>&\ 2x_2 - p(x_2)\\
=&\ v_2(x_2) - p(x_2)
\end{align*}
Thus agent 2 would rather purchase $x_1$ than $x_2$, so $x_2 \not\in D_2(p)$. Therefore $(\x, p)$ is not a WE.
\end{proof}

For $\rho = 0$, the situation is different. Recall that Fisher market equilibrium always maximizes Nash welfare, and we can simulate Fisher market budgets by setting
\[
p(x_i) = \begin{cases}
0 &\ \text{ if } \sum_{j \in M} q_j x_{ij} \le 1\\
\infty &\ \text{ otherwise}
\end{cases}
\]
where $q_1\dots q_m$ are the optimal Lagrange multipliers in the convex program for maximizing Nash welfare. Gale and Eisenberg's famous result implies that for such a pricing rule, a WE always exists, and all WE maximize Nash welfare~\cite{Eisenberg1961,Eisenberg1959}. Note that for $\sum_{j \in M} q_j x_{ij} > 1$, $p(x_i) = \infty$ can be implemented by setting $\mfrac{\partial p(x_i)}{\partial x_{ij}}$ to be at least $\max_{i \in N} \max_{x_i \in [0,1]^m} \mfrac{\partial v_i(x_i)}{\partial x_{ij}}$. This ensures that no agent purchases a bundle $x_i$ such that $\sum_{j \in M} q_j x_{ij} > 1$.

The above pricing rule is somewhat artificial, however. One natural question is whether Nash welfare maximization can be implemented with a differentiable pricing rule. We next show that the answer is no.

\begin{theorem}\label{thm:nash-counter}
Consider the instance with $n=2$, $m=1$, $v_1(x) = x$ and $v_2(x) = 2x$. Then there is no allocation $\x$ maximizing Nash welfare and differentiable pricing rule $p$ such that $(\x, p)$ is a WE.
\end{theorem}

\begin{proof}
Suppose the opposite: that such $\x, p$ exist. The unique $\x$ maximizing Nash welfare must have $x_1 = x_2 = 1/2$. Since $p, v_1,$ and $v_2$ are all differentiable, we have $x_i \in D_i(p)$ if and only if $\mfrac{\dif p(x_i)}{\dif x_i} = \mfrac{\dif v_i(x_i)}{\dif x_i}$. Since $x_1 = x_2$, we have $\mfrac{\dif p(x_1)}{\dif x_1} = \mfrac{\dif p(x_2)}{\dif x_2}$. Thus implies $\mfrac{\dif v_1(x_1)}{\dif x_1} = \mfrac{\dif v_2(x_2)}{\dif x_2}$, which is a contradiction. We conclude that no such $\x, p$ exist.
\end{proof}

\section{Conclusion}

In this paper, we studied a simple family of convex pricing rules, motivated by the widespread use of convex pricing in the real world, especially for water. We proved that these pricing rules implement CES welfare maximization in Walrasian equilibrium, providing a formal quantitative interpretation of the frequent informal claim that convex pricing promotes equality. Furthermore, by tweaking the exponent of the pricing rule, the social planner can precisely control the tradeoff between equality and efficiency. This result also shows that convex pricing is not necessarily economically inefficient, as often claimed; it simply maximizes a different welfare function than the traditional utilitarian one.


Improved implementation is perhaps the most important of the future directions we propose. One concrete possibility is a \emph{\tat}: an iterative algorithm where on each step, each agent reports her demand for the current pricing rule, and the pricing rule is adjusted accordingly. Demand queries are arguably easier for agents to answer than valuation gradient queries. Some implementation questions -- in particular, how to deal with Sybil attacks -- would likely need to be handled on a case-by-case basis.

Aside from the implementation itself, there is the additional challenge of convincing market designers to consider using this type of convex pricing. Equality is generally thought to be desirable, but sellers may be concerned that this will decrease their revenue. In future work, we hope to show that our pricing rule guarantees a good approximation of the optimal revenue for sellers.

Another possible direction would be CES welfare maximization for indivisible goods. The analogous pricing rule would be $p(S) = (\sum_{j \in S} q_j)^{1/\rho}$, where $S$ is a set of indivisible goods. It seems like very different theoretical techniques would be needed in this setting (along with perhaps a gross substitutes assumption), but we suspect that the same intuition of convex pricing improving equality would hold.



\bibliographystyle{plain}
\bibliography{refs}


\appendix

\section{Connections to Fisher markets}\label{sec:fisher}

The focus of this paper is on markets for quasilinear utilities, where agents can spend as much money as they want, and the amount spent is incorporated into their resulting utility. The other predominant market model assumes each agent $i$ has a finite budget $B_i$ of money to spend, and has no value for leftover money (in general, this implies that each agent $i$ spends exactly $B_i$). This is called the \emph{Fisher market} model.\footnote{There are also more general versions of this model that allow each agent's initial endowment to be goods instead of money (``exchange economies") and/or allow production (``Arrow-Debreu markets").} In this section, we explore connections between our results and the Fisher market model.

In the Fisher market model, each agent's utility $u_i(x_i)$ is simply $v_i(x_i)$. For pricing rule $p$, the Fisher market demand set is given by
\[
D_i^F(p) = \argmax_{x_i \in \bbrpos^m:\ p(x_i) \le B_i} v_i(x_i)
\]
We will reserve the notation $D_i(p)$ for the demand set in the quasilinear case, i.e., $D_i(p) = \argmax_{x_i \in \bbrpos^m} (v_i(x_i) - p(x_i))$.

For an allocation $\x$, agent budgets $\bigB = (B_1,\dots,B_n)$, and a pricing rule $p$, $(\x, \bigB, p)$ is a \emph{Fisher market Walrasian equilibrium} if (1) $x_i \in D_i^F(p)$ for all $i \in N$, and (2) $\sum_{i \in N} x_{ij} \le 1$ for all $j \in M$, and if good $j$ has nonzero cost, $\sum_{i \in N} x_{ij} = 1$.\footnote{Recall that good $j$ has nonzero cost for $j$ if there is a bundle $x_i$ such that $x_{i\ell} = 0$ for all $\ell \ne j$, but $p(x_i) > 0$.} These are the same two conditions for Walrasian equilibrium in quasilinear markets: the only change is the definition of the demand set. To distinguish, we will use the terms ``Fisher WE" and ``quasilinear WE".

\subsection{CES welfare maximization in Fisher markets}

In the quasilinear model, agents can express not only their relative values between goods, but also the absolute scale of their valuation (i.e., the ``intensity" of their preferences) by choosing how much money to spend. In contrast, agents in the Fisher market model spend exactly their budget, and so have no way to express the absolute scale of their valuation. This should make us pessimistic about the possibility of CES welfare maximization in the Fisher market model in general. Indeed, consider a single good and two agents with valuations $v_1(x) = x$, $v_2(x) = 2x$. For any $\rho > 0$, any optimal allocation $\x \in \Psi(\rho)$ has $x_2 > x_1$. But if $B_1 = B_2$, any Fisher market Walrasian equilibrium will always have $x_1 = x_2$, since both agents simply spend their entire budget on the single good.

However, in general we can convert a quasilinear WE to a Fisher WE if the agents' budgets are sized appropriately. Specifically, we need agent $i$'s budget to be exactly the amount she pays in the quasilinear WE:

\begin{theorem}\label{thm:fisher-we}
Suppose $(\x, p)$ is a quasilinear WE, and let $B_i = p(x_i)$. Then $(\x, \bigB, p)$ is a Fisher WE.
\end{theorem}

\begin{proof}
For all $i \in N$, $x_i$ is affordable to agent $i$ under $p$ by definition of $B_i$. Suppose there were another bundle $y_i$ such that $p(y_i) \le B_i$ but $v_i(y_i) > v_i(x_i)$. That would contradict $x_i \in D_i(p)$ for the quasilinear case, since $u_i(y_i) = v_i(y_i) - p(y_i) > v_i(x_i) - p(y_i) \ge v_i(x_i) - p(x_i) = u_i(x_i)$. Therefore $x_i \in D_i^F(p)$ for all $i \in N$. Furthermore, the market clearing conditions for Fisher WE and quasilinear WE are identical. We conclude that $(\x, \bigB, p)$ is a Fisher WE.
\end{proof}

Combining the above result with Theorem~\ref{thm:main} gives us the following corollary for CES welfare maximization:

\begin{corollary}\label{cor:fisher-we}
Assume each $v_i$ is homogeneous of degree $r$, concave, and differentiable. Let $\rho \in (0, 1]$, and $p(x_i) = \rho r^{\frac{\rho -1}{\rho}} (\sum_{j \in M} q_j x_{ij})^{1/\rho}$, where $q_1,\dots, q_m$ are optimal Lagrange multipliers for Program~\ref{pro:ces}. For $\x \in \Psi(\rho)$, let $B_i = p(x_i)$ for all $i \in N$, Then $(\x, \bigB, p)$ is a Fisher WE.
\end{corollary}

Perhaps the more interesting connection relates to the welfare function being optimized. In the case of linear pricing, the Fisher market Walrasian equilibria are exactly the budget-weighted maximum Nash welfare allocations.\footnote{Recall that Nash welfare corresponds to $\rho = 0$, and the budget-weighted Nash welfare of an allocation $\x$ is $\prod_{i \in N} v_i(x_i)^{B_i}$.} One natural question is whether the Fisher market equilibria from Theorem~\ref{thm:fisher-we} also optimize a budget-weighted CES welfare function. We answer this in the affirmative. Recall that we define $\Phi_\bigB(\rho, \x) = \big(\sum_{i \in N} B_i v_i(x_i)^{\rho}\big)^{1/\rho}$, and $\Psi_\bigB(\rho) = \argmax_\x \Phi_\bigB(\rho, \x)$.

\begin{lemma}\label{lem:fisher-double-opt}
Assume each $v_i$ is concave and differentiable. Let $\xprime$ be any allocation, let $a_i = v_i(x_i')$ for each $i \in N$, and let $\rho \in (0, 1]$. Then $\xprime \in \Psi(\rho)$ if and only if $\xprime \in \Psi_\A(\rho - 1)$.
\end{lemma}

\begin{proof}
When $\rho = 1$, $\rho - 1 = 0$, so Program~\ref{pro:ces} does not apply, and we must handle this case separately. We first consider $\rho \ne 1$. The Lagrangian for Program~\ref{pro:ces} for $\Psi_\A(\rho-1)$ is $L(\x, \q) = \frac{1}{\rho-1} \sum_{i \in N} a_i v_i(x_i)^{\rho-1} - \sum_{j \in M} q_j (\sum_{i \in N} x_{ij} - 1)$. The KKT conditions imply that $\x \in \Psi_\A(\rho-1)$ if and only if there exist Lagrange multipliers $q_1,\dots, q_m$ such that:
\begin{enumerate}
\item Stationarity: $\mfrac{\partial L(\x, \q)}{\partial x_{ij}} = a_i v_i(x_i)^{\rho - 2} \mfrac{\partial v_i(x_i)}{\partial x_{ij}} \le 0$ for all $i,j$.\footnote{Note that since $\xprime$ is a fixed allocation, $a_i$ is just some constant, so the differentiation does not affect it.} Furthermore, if $x_{ij} > 0$, the inequality holds with equality.
\item Complementary slackness: for all $j \in M$, either $\sum_{i \in N} x_{ij} = 1$, or $q_j = 0$.
\end{enumerate}

For Program~\ref{pro:ces} for $\Psi(\rho)$, as before we have $L'(\x, \q) = \frac{1}{\rho} \sum_{i \in N} v_i(x_i)^\rho - \sum_{j \in M} q_j (\sum_{i \in N} x_{ij} - 1)$. Thus the KKT conditions imply that $\x \in \Psi(\rho)$ if and only if there exist $q_1',\dots,q_m' \in \bbrpos$ such that (1) $v_i(x_i)^{\rho - 1} \mfrac{\partial v_i(x_i)}{\partial x_{ij}} \le q_j$ for all $i,j$, and when $x_{ij} > 0$, the inequality holds with equality, and (2) for all $j \in M$, either $\sum_{i \in N} x_{ij} = 1$, or $q_j = 0$. Note that if $q_j = q_j'$ for all $j \in M$, the complementary slackness conditions become equivalent.


Next, for $\x = \xprime$ we have
\[
v_i(x_i')^{\rho - 1} \frac{\partial v_i(x_i')}{\partial x'_{ij}} = v_i(x'_i) v_i(x'_i)^{\rho - 2} \frac{\partial v_i(x'_i)}{\partial x'_{ij}} = a_i v_i(x'_i)^{\rho - 2} \frac{\partial v_i(x'_i)}{\partial x'_{ij}}
\]
Therefore for given $q_j$, we have $q_j \ge v_i(x'_i)^{\rho - 1} \mfrac{\partial v_i(x'_i)}{\partial x'_{ij}}$ if and only if $q_j \ge a_i v_i(x'_i)^{\rho - 2} \mfrac{\partial v_i(x'_i)}{\partial x'_{ij}}$, and $q_j = v_i(x'_i)^{\rho - 1} \mfrac{\partial v_i(x'_i)}{\partial x'_{ij}}$ if and only if $q_j \ge a_i v_i(x'_i)^{\rho - 2} \mfrac{\partial v_i(x'_i)}{\partial x'_{ij}}$.

Now suppose $\xprime \in \Psi(\rho)$. Then there exist $q_1,\dots, q_m \in \bbrpos$ that satisfy both stationarity and complementary slackness. Then as we showed above, $\xprime$ and $q_1,\dots, q_m$ satisfy stationarity for $\Psi_\A(\rho-1)$. Furthermore, the complementary slackness conditions are equivalent, so we have $\xprime \in \Psi_\A(\rho-1)$.

Similarly, suppose $\xprime \in \Psi_\A(\rho-1)$. Then there exist $q_1,\dots, q_m$ satisfying stationarity and complementary slackness, so the same $q_1,\dots, q_m$ along with $\xprime$ satisfy the KKT conditions for $\Psi(\rho)$. Therefore $\Psi(\rho)$, and we conclude that $\xprime \in \Psi(\rho)$ if and only if $\xprime \in \Psi_\A(\rho-1)$ for $\rho \ne 1$.

All of the above was for $\rho \ne 1$; it remains to handle the case of $\rho = 1$. In this case, we can use the same KKT conditions for $\Psi(\rho)$, but must use a different convex program for $\Psi_\A(\rho-1)$. Consider the following convex program for maximizing Nash welfare (i.e., CES welfare for $\rho = 0$):
\begin{alignat}{2}
\max\limits_{\x \in \bbrpos^{n\times m}} &\ \sum_{i \in N} a_i \log v_i(x_i) \label{pro:eg} \\ 
s.t.\ &\ \sum\limits_{i \in N} x_{ij}\leq 1\quad &&\ \forall j \in M  \nonumber
\end{alignat}
This is known as the Eisenberg-Gale program~\cite{Eisenberg1961,Eisenberg1959}. In this case, the stationarity condition requires that $\mfrac{\partial}{\partial x_{ij}} a_i \log v_i(x_i) = a_i v_i(x_i)^{-1} \mfrac{\partial v_i(x_i)}{\partial x_{ij}} \le q_j$ for all $i,j$, and when $x_{ij} > 0$, the inequality holds with equality. Since $\rho = 1$ here, we have $\rho - 2 = - 1$. Thus the stationarity condition for $\Psi_\A(\rho - 1)$ requires that $a v_i(x_i)^{\rho-2} \mfrac{\partial v_i(x_i)}{\partial x_{ij}} \le q_j$ for all $i,j$ (and if $x_{ij} > 0$, this holds with equality). This is exactly what we had above, and since we are using the same KKT conditions for $\Psi(\rho)$, this case reduces to the case for $\rho \ne 1$. Therefore for $\rho = 1$, $\xprime \in \Psi(\rho)$ if and only if $\xprime \in \Psi_\A(\rho-1)$.
\end{proof}

Combining Theorem~\ref{thm:fisher-we} and Lemma~\ref{lem:fisher-double-opt}, we get:
\begin{theorem}\label{thm:full}
Assume each $v_i$ is homogeneous of degree $r$, concave, and differentiable, let $\rho \in (0, 1]$, let $q_1,\dots,q_m \in \bbrpos$, and let $p(x_i) = \rho \big(\sum_{j \in M} q_j x_{ij}\big)^{1/\rho}$. Given $\x \in \Psi(\rho)$, let $B_i = p(x_i)$. Then all of the following hold:
\begin{enumerate}
\item $(\x, p)$ is a quasilinear WE.
\item $(\x, \bigB, p)$ is a Fisher WE.
\item $\x \in \Psi_\bigB(\rho - 1)$
\end{enumerate}
\end{theorem}

\begin{proof}
The first and second conditions hold by Theorems~\ref{thm:main} and \ref{thm:fisher-we}, respectively. Then Corollary~\ref{cor:utility-price} implies that $p(x_i) = r\rho v_i(x_i)$. Let $a_i = v_i(x_i) = \mfrac{B_i}{r\rho}$. Thus by Lemma~\ref{lem:fisher-double-opt}, we have $\x \in \Psi_\A(\rho-1)$. Since scaling all agents' multipliers by the same factor does not affect $\Psi_\A(\rho-1)$, we have $\x \in \Psi_{r\rho \A}(\rho-1) = \Psi_\bigB(\rho-1)$, as required.
\end{proof}

It is worth noting that the special case of Theorem~\ref{thm:fisher-we} for $\rho = 1$ and Leontief utilities with $w_{ij} \in \{0,1\}$\footnote{This is also known as the bandwidth allocation setting, where each good represents a link in a network, and agent $i$ has $w_{ij} = 1$ for every link $j$ on a fixed path (and $w_{ij} = 0$ otherwise).} is implied by the work of Kelly et al.~\cite{Kelly1998}.

\section{CES welfare maximization for Leontief valuations}\label{sec:leontief}

We say that $v_i$ is \emph{Leontief} if there exist weights $w_1,\dots,w_m \in \bbrpos$ such that
\[
v_i(x_i) = \min_{j:\ w_{ij} \ne 0} \frac{x_{ij}}{w_{ij}}
\]
Leontief valuations are not differentiable, and so Theorem~\ref{thm:main} does not apply. In this section, we handle Leontief valuations as a special case. Although there are many non-differentiable valuations we could consider, there is substantial related work on Leontief valuations (\cite{Goel2019nash,Plaut2019}), so we find it worthwhile to show that our result does indeed extend to this case.

Recall Program~\ref{pro:ces}:
\begin{alignat*}{2}
\max\limits_{\x \in \bbrpos^{n\times m}} &\ \frac{1}{\rho}\sum_{i \in N} v_i(x_i)^\rho \\
s.t.\ &\ \sum\limits_{i \in N} x_{ij}\leq 1\quad &&\ \forall j \in M  \nonumber
\end{alignat*}

We will work with a specialized version of this for Leontief utilities:
\begin{alignat}{2}
\max\limits_{\x \in \bbrpos^{n\times m}, \bfa \in \bbrpos^m} \frac{1}{\rho}\sum_{i \in N} \alpha_i^\rho &\ \label{pro:leontief}\\
s.t.\ w_{ij} \alpha_i \le&\ x_{ij}\quad &&\ \forall i \in N, j \in M  \nonumber\\
\sum\limits_{i \in N} x_{ij}\le&\ 1\quad &&\ \forall j \in M  \nonumber
\end{alignat}
where we use $\bfa$ to denote the vector $(\alpha_1,\dots,\alpha_n) \in \bbrpos^n$.

Also recall each agent's demand set $D_i(p) = \argmax_{x_i \in \bbrpos^m}\ \big( v_i(x_i) - p(x_i)\big)$. Similarly to Program~\ref{pro:leontief}, we consider the following equivalent (specialized) convex program for agent $i$'s demand set:
\begin{alignat}{2}
\max\limits_{x_i \in \bbrpos^m, \alpha_i \in \bbrpos} \big(\alpha_i - p(x_i)\big) &\ \label{pro:leontief-demand}\\
s.t.\ w_{ij} \alpha_i \le&\ x_{ij}\quad &&\ \forall j \in M  \nonumber
\end{alignat}

\begin{theorem}\label{thm:leontief}
Assume each $v_i$ is Leontief with weights $w_{i1},\dots,w_{im}$. Then for any $\rho \in (0, 1]$ and any allocation $\x$, we have $\x \in \Psi(\rho)$ if and only if there exist $q_1,\dots, q_m\in \bbrpos$ such that for the pricing rule $p(x_i) = \rho (\sum_{j \in M} q_j x_{ij})^{1/\rho}$, $(\x, p)$ is a WE. Furthermore, $q_1,\dots, q_m$ are optimal Lagrange multipliers (for the $\sum_{i \in N} x_{ij} \le 1$ constraints) for Program~\ref{pro:leontief}.
\end{theorem}

\begin{proof}
We first claim that in an optimal solution $\x,\bfa$ to either Program~\ref{pro:leontief} or Program~\ref{pro:leontief-demand}, we have $v_i(x_i) = \alpha_i$ for all $i \in N$: that is, that these programs are doing what we want them to. To see this, note that $\alpha_i \le x_{ij} / w_{ij}$ for all $j$ with $w_{ij} \ne 0$, so $\alpha_i \le v_i(x_i)$. Furthermore, at least one constraint involving $\alpha_i$ must be tight: otherwise, we could increase $\alpha_i$ and thus the objective value. In particular, we must have $\alpha_i = \min_{j:\ w_{ij} \ne 0} \frac{x_{ij}}{w_{ij}} = v_i(x_i)$. Thus Program~\ref{pro:leontief-demand} is indeed maximizing $v_i(x_i) - p(x_i)$, so $x_i \in D_i(p)$ if and only if $(x_i, \alpha_i)$ is optimal for Program~\ref{pro:leontief-demand} (for some $\alpha_i$). Similarly, Program~\ref{pro:leontief} is indeed maximizing $\frac{1}{\rho}\sum_{i \in N} v_i(x_i)^\rho$ subject to $\sum_{i \in N} x_{ij} \le 1$ for all $j \in M$, so $(\x, \bfa)$ is optimal for Program~\ref{pro:leontief} if and only if $\x$ is optimal for Program~\ref{pro:ces}. Therefore $\x \in \Psi(\rho)$ if and only if $(\x, \bfa)$ is optimal for Program~\ref{pro:leontief} (for some $\bfa$). 

Next, we write the Lagrangian of Program~\ref{pro:leontief}:\footnote{As in the proof of Theorem~\ref{thm:main}, we omit the $\x \in \bbrpos^{m\times n}$ constraint from the Lagrangian incorporate it into the KKT conditions instead.}
\[
L(\x, \bfa, \q, \bflam) = \frac{1}{\rho} \sum_{i \in N} \alpha_i^\rho - \sum_{i \in N} \sum_{j \in M} \lambda_{ij} (w_{ij} \alpha_i - x_{ij}) - \sum_{j \in M} q_j \Big(\sum_{i \in N} x_{ij} - 1\Big)
\]
We have strong duality by Slater's condition, so the KKT conditions are both necessary and sufficient for optimality. That is, $(\x, \bfa)$ is optimal if and only if there exist $\q \in \bbrpos^m$, $\bflam \in \bbrpos^{m\times n}$ such that all of the following hold:\footnote{As in the proof of Theorem~\ref{thm:main}, primal and dual feasibility are trivially satisfied.}
\begin{enumerate}
\item Stationarity for $\x$: $\mfrac{\partial L(\x, \bfa, \q, \bflam)}{\partial x_{ij}} \le 0$ for all $i,j$. Furthermore, if $x_{ij} > 0$, the inequality holds with equality.
\item Stationarity for $\bfa$: $\mfrac{\partial L(\x, \bfa, \q, \bflam)}{\partial \alpha_i} \le 0$ for all $i \in N$. Furthermore, if $\alpha_i > 0$, the inequality holds with equality.
\item Complementary slackness for $\q$: for all $j \in M$, either $\sum_{i \in N} x_{ij} = 1$, or $q_j = 0$.
\item Complementary slackness for $\bflam$: for all $i \in N$, $j\in M$, either $w_{ij} \alpha_i = x_{ij}$ or $\lambda_{ij} = 0$.
\end{enumerate}

Similarly, let $L_i'$ denote the Lagrangian of Program~\ref{pro:leontief-demand} for agent $i$:
\[
L_i'(x_i, \alpha_i, \lambda_i) = \alpha_i - p(x_i) - \sum_{j \in M} \lambda_{ij} (w_{ij} \alpha_i - x_{ij})
\]
where $\lambda_i = (\lambda_{i1},\dots,\lambda_{im}) \in \bbrpos^m$.
We again have strong duality, so the KKT conditions are again necessary and sufficient. Let $L_i'(x_i, \alpha_i, \lambda_i)$ denote the Lagrangian of this program; then $(x_i, \alpha_i)$ is optimal for Program~\ref{pro:leontief-demand} if and only if all of the following hold:
\begin{enumerate} 
\item Stationarity for $x_i$: $\mfrac{\partial L_i'(x_i, \alpha_i, \lambda_i)}{\partial x_{ij}} \le 0$ for all $j \in M$. If $x_{ij} > 0$, the inequality holds with equality.
\item Stationarity for $\alpha_i$: $\mfrac{\partial L_i'(x_i, \alpha_i, \lambda_i)}{\partial \alpha_i} \le 0$. If $\alpha_i > 0$, the inequality holds with equality.
\item Complementary slackness for $\lambda_i$: for all $i \in N$, $j\in M$, either $w_{ij} \alpha_i = x_{ij}$ or $\lambda_{ij} = 0$.
\end{enumerate}

We will claim that $(\x, \bfa, \q, \bflam)$ is optimal for Program~\ref{pro:leontief} if and only if for all $i \in N$, $(x_i, \alpha_i, \alpha_i^{1-\rho} \lambda_i)$ is optimal for Program~\ref{pro:leontief-demand}. Essentially, we show that if complementary slackness holds (for either program), the stationarity conditions are equivalent. To begin, we can explicitly compute the relevant partial derivatives for given $\x, \bfa, \q, \bflam$, with $p(x_i) = \rho (\sum_{j \in M} q_j x_{ij})^{1/\rho}$:
\begin{align*}
\frac{\partial L(\x, \bfa, \q, \bflam)}{\partial x_{ij}} =&\ \lambda_{ij} - q_j \\
\frac{\partial L(\x, \bfa, \q, \bflam)}{\partial \alpha_i} =&\ \alpha_i^{\rho - 1} - \sum_{j \in M} \lambda_{ij} w_{ij}\\
\frac{\partial L_i'(x_i, \alpha_i, \lambda'_i)}{\partial x_{ij}} =&\ \lambda'_{ij} - q_j \Big(\sum_{\ell \in M} q_\ell x_{i\ell}\Big)^{\frac{1-\rho}{\rho}}\\
\frac{\partial L_i'(x_i, \alpha_i, \lambda'_i)}{\partial \alpha_i} =&\ 1 - \sum_{j \in M} \lambda'_{ij} w_{ij}
\end{align*}

\textbf{Part 1:} $(\implies)$ Suppose that $\x \in \Psi(\rho)$. Then there exist $\bfa, \q, \bflam$ such that the KKT conditions for Program~\ref{pro:leontief} are satisfied for $(\x, \bfa, \q, \bflam)$. We first claim that $\alpha_i > 0$ for all $i \in N$. Suppose not: stationarity implies that $\alpha_i^{\rho-1} \le \sum_{j \in M} \lambda_{ij} w_{ij}$, but since $\rho - 1< 0$, the left hand side is not defined for $\alpha_i = 0$. Thus $\alpha_i > 0$.

Therefore by stationarity for $\alpha_i$, we have $\alpha_i^{\rho-1} = \sum_{j \in M} \lambda_{ij} w_{ij}$. Let $\lambda'_{ij} = \alpha_i^{1-\rho} \lambda_{ij}$ for all $i,j$\footnote{Note that this is \emph{not} defining $\lambda'_{ij}$ to be a function of $\alpha_i$. This is defining $\lambda'_{ij}$ based on a fixed value of $\alpha_i$: in particular, the value from $(\x, \bfa, \q, \bflam)$, which we assumed to be optimal for Program~\ref{pro:leontief}. Consequently, the derivatives in the KKT conditions treat $\lambda'_{ij}$ as a constant.}. Then $\alpha_i^{\rho-1} = \sum_{j \in M} \lambda_{ij} w_{ij}$ is equivalent to $1 = \sum_{j \in M} \lambda'_{ij} w_{ij}$, and thus $\mfrac{\partial L_i'(x_i, \alpha_i, \lambda'_i)}{\partial \alpha_i} = 0$ for all $i \in N$. Thus for all $i \in N$, $(x_i, \alpha_i, \lambda'_i)$ satisfies stationarity for $\alpha_i$ for Program~\ref{pro:leontief-demand}.

We now turn to the $x_{ij}$ variables. Stationarity for $x_{ij}$ in Program~\ref{pro:leontief} implies that $\lambda_{ij} = q_j$ whenever $x_{ij} > 0$. Furthermore, complementary slackness for $\lambda_{ij}$ implies that if $\lambda_{ij} > 0$, $w_{ij} \alpha_i = x_{ij}$. Thus whenever $q_j > 0$ and $x_{ij} > 0$, $w_{ij} \alpha_i = x_{ij}$ and $\lambda_{ij} = q_j$. Therefore for all $i,j$,
\begin{align*}
\frac{\partial L_i'(x_i, \alpha_i, \lambda'_i)}{\partial x_{ij}} =&\ \lambda'_{ij} - q_j \Big(\sum_{\ell: q_\ell, x_{i\ell} > 0} q_\ell x_{i\ell} \Big)^{\frac{1-\rho}{\rho}}\\
=&\ \lambda'_{ij} - q_j \Big(\sum_{\ell \in M} \lambda_{i\ell} w_{i\ell} \alpha_i\Big)^{\frac{1-\rho}{\rho}}\\
=&\ \lambda'_{ij} - q_j \Big(\alpha_i \sum_{\ell \in M} \lambda_{i\ell} w_{i\ell} \Big)^{\frac{1-\rho}{\rho}}\\
=&\ \lambda'_{ij} - q_j (\alpha_i \alpha_i^{\rho-1})^{\frac{1-\rho}{\rho}}\\
=&\ \alpha_i^{1-\rho} \lambda_{ij} - q_j \alpha_i^{1-\rho}\\
=&\ \alpha_i^{1-\rho} \frac{\partial L(\x, \bfa, \q, \bflam)}{\partial x_{ij}}
\end{align*}
We have $\mfrac{\partial L(\x, \bfa, \q, \bflam)}{\partial x_{ij}} \le 0$ for all $i,j$ by stationarity (and the inequality holds with equality when $x_{ij} > 0$), so $\frac{\partial L_i'(x_i, \alpha_i, \lambda'_i)}{\partial x_{ij}} \le 0$ for all $j \in M$ (and the inequality holds with equality when $x_{ij} > 0$). Thus for each $i \in N$, $(x_i, \alpha_i, \lambda'_i)$ satisfies stationarity for Program~\ref{pro:leontief-demand} for $x_{ij}$ for all $j \in M$.

As mentioned above, we have $w_{ij} \alpha_i = x_{ij}$ whenever $\lambda_{ij} > 0$. Since $\lambda'_{ij} > 0$ if and only if $\lambda_{ij} > 0$, we have $w_{ij} \alpha_i = x_{ij}$ whenever $\lambda'_{ij} > 0$. Thus for each $i \in N$, $(x_i, \alpha_i, \lambda'_i)$ satisfies complementary slackness for Program~\ref{pro:leontief-demand}. Therefore $(x_i, \alpha_i, \lambda'_i)$ satisfies the KKT conditions, and thus is optimal for Program~\ref{pro:leontief-demand}. Therefore $x_i \in D_i(p)$ for all $i \in N$. The complementary slackness condition for $\q$ is identical to the market clearing condition, so we conclude that $(\x, p)$ is a WE.

\textbf{Part 2:} $(\impliedby)$ Suppose that $(\x, p)$ is a WE, where $p(x_i) = \rho (\sum_{j \in M} q_j x_{ij})^{1/\rho}$ for constants $q_1,\dots,q_m \in \bbrpos$. Then $x_i \in D_i(p)$ for all $i \in N$, so there exists $\bfa, \bflamp$ such that $(x_i, \alpha_i, \lambda'_i)$ is optimal for Program~\ref{pro:leontief-demand} for all $i \in N$.

Thus by stationarity, we have $\frac{\partial L_i'(x_i, \alpha_i, \lambda'_i)}{\partial \alpha_i} \le 0$ and $\frac{\partial L_i'(x_i, \alpha_i, \lambda'_i)}{\partial x_{ij}} \le 0$ for all $i,j$ (and if $\alpha_i > 0$ and $x_{ij} > 0$, the inequalities hold with equality). Using the definition of $p$, we have $\frac{\partial L_i'(x_i, \alpha_i, \lambda'_i)}{\partial x_{ij}} = \lambda'_{ij} - q_j \big(p(x_i)/\rho\big)^{1-\rho}$. Thus $1 \le \sum_{j \in M} \lambda'_{ij} w_{ij}$ and $\lambda'_{ij} \le q_j \big(p(x_i)/\rho\big)^{1-\rho}$. 

We first claim that $\alpha_i > 0$ for all $i \in N$. For each agent $i$, there must exist $j \in M$ such that $\lambda'_{ij} > 0$ and $w_{ij} > 0$: otherwise $1 \le \sum_{j \in M} \lambda'_{ij} w_{ij}$ would be impossible. Consider any such $j$: then $0 < \lambda'_{ij} \le q_j \big(p(x_i)/\rho\big)^{1-\rho}$, so we must have $p(x_i) > 0$. Suppose $\alpha_i = 0$: then the optimal objective value of Program~\ref{pro:leontief-demand} is $\alpha_i - p(x_i) < 0$. But setting $x_{ij} = 0$ for all $j \in M$ achieves an objective value of 0, so $\alpha_i - p(x_i) < 0$ cannot be optimal. This is a contradiction, and so $\alpha_i > 0$ for all $i \in N$.

Returning to the stationarity conditions, we then have $1 = \sum_{j \in M}\lambda'_{ij} w_{ij}$. Complementary slackness implies that $w_{ij} \alpha_i = x_{ij}$ whenever $\lambda'_{ij} > 0$, so we get
\begin{align*}
\alpha_i =&\ \sum_{j \in M} \lambda'_{ij} w_{ij} \alpha_i \\
\alpha_i =&\ \sum_{j: \lambda'_{ij} > 0} \lambda'_{ij} w_{ij} \alpha_i \\
\alpha_i =&\ \sum_{j \in M} \lambda'_{ij} x_{ij}
\end{align*}
Combining this with $\lambda'_{ij} =  q_j \big(p(x_i)/\rho\big)^{1-\rho}$ whenever $x_{ij} > 0$ gives us
\begin{align*}
\alpha_i =&\ \sum_{j: x_{ij} > 0} \lambda'_{ij} x_{ij}\\
=&\ \sum_{j \in M} q_j x_{ij} \big(p(x_i)/\rho\big)^{1-\rho}\\
=&\ \big(p(x_i)/\rho\big)^{1-\rho} \sum_{j \in M} q_j x_{ij} \\
=&\ \big(p(x_i)/\rho\big)^{1-\rho} \big(p(x_i)/\rho\big)^\rho \\
=&\ p(x_i)/\rho
\end{align*}

Now let $\lambda_{ij} = \alpha_i^{\rho - 1} \lambda'_{ij}$ for all $i,j$. We claim that $(\x, \bfa, \q, \bflam)$ satisfies the KKT conditions for Program~\ref{pro:leontief}. For each $(i,j)$ pair, we have
\[
\frac{\partial L(\x, \bfa, \q, \bflam)}{\partial \alpha_i} 
= \alpha_i^{\rho - 1} - \sum_{j \in M} \lambda_{ij} w_{ij} 
= \alpha_i^{\rho - 1}\Big(1 - \sum_{j \in M} \lambda'_{ij} w_{ij}\Big) = \alpha_i^{\rho - 1} \frac{\partial L_i'(x_i, \alpha_i, \lambda'_i)}{\partial \alpha_i}
\]
Since $\alpha_i > 0$, stationarity for Program~\ref{pro:leontief-demand} implies that $\frac{\partial L_i'(x_i, \alpha_i, \lambda'_i)}{\partial \alpha_i} = 0$, so we have $\frac{\partial L(\x, \bfa, \q, \bflam)}{\partial \alpha_i} = 0$. Next, we have
\begin{align*}
\frac{\partial L(\x, \bfa, \q, \bflam)}{\partial x_{ij}} =&\ \lambda_{ij} - q_j\\
=&\ \alpha_i^{\rho - 1} \lambda'_{ij} - q_j\\
=&\ \alpha_i^{\rho - 1} ( \lambda'_{ij} - q_j \alpha_i^{1-\rho})\\
=&\ \alpha_i^{\rho - 1} \Big(\lambda'_{ij} - q_j \big(p(x_i)/\rho\big)^{1-\rho}\Big)\\
=&\ \alpha_i^{\rho - 1} \frac{\partial L_i'(x_i, \alpha_i, \lambda'_i)}{\partial x_{ij}}
\end{align*}
Stationarity for Program~\ref{pro:leontief-demand} implies that $\frac{\partial L_i'(x_i, \alpha_i, \lambda'_i)}{\partial x_{ij}} \le 0$ for all $i,j$ (and when $x_{ij} > 0$, this holds with equality), so we have $\frac{\partial L(\x, \bfa, \q, \bflam)}{\partial x_{ij}}$ for all $i,j$ (and when $x_{ij} > 0$, this holds with equality). Thus we have shown that $(\x, \bfa, \q, \bflam)$ satisfies stationarity for Program~\ref{pro:leontief}. As before, the market clearing condition is equivalent to complementary slackness for $\q$. By complementary slackness for $\bflamp$ (for Program~\ref{pro:leontief-demand}), we have $w_{ij} \alpha_i = x_{ij}$ whenever $\lambda'_{ij} > 0$. By definition, $\lambda'_{ij} > 0$ if and only if $\lambda_{ij} > 0$, so this implies the required complementary slackness for $\bflam$ (for Program~\ref{pro:leontief}). Therefore $(\x, \bfa, \q, \bflam)$ satisfies the KKT conditions for Program~\ref{pro:leontief}, and thus is optimal for that program. We conclude that $\x \in \Psi(\rho)$.

\end{proof}

\section{The First Welfare Theorem and linear pricing}\label{sec:linear}

Recall our main result:

\mainThm*

For this class of valuations, Theorem~\ref{thm:main} for $\rho = 1$ implies the First Welfare Theorem: $p$ becomes a linear pricing rule, and CES welfare for $\rho = 1$ is just utilitarian welfare. In particular, Theorem~\ref{thm:main} implies both the existence of a WE, and that every WE maximizes utilitarian welfare.

Typically, the ``the First Welfare Theorem" refers to just half of this: that every WE maximizes utilitarian welfare. The reason is that WE are not always guaranteed to exist: for divisible goods, generally at least concavity or quasi-concavity of valuations is necessary. On the other hand, very few assumptions are needed to show that linear pricing equilibria always maximize utilitarian welfare; for example, divisibility of goods is not needed. We provide a proof of this below. 


\begin{theorem}[The First Welfare Theorem]\label{thm:first-welfare}
Let $\X_i \subset \bbr^m$ denote the set of feasible bundles for agent $i$ (not necessarily convex, and not necessarily the same for all agents). Let $D_i(p) = \argmax_{x_i \in \X_i} (v_i(x_i) - p(x_i))$ and assume $p$ is linear. Then if $(\x, p)$ is a WE, $\x$ maximizes utilitarian welfare.
\end{theorem}

\begin{proof}
Since $p$ is linear, there exist $q_1,\dots,q_m$ such that $p(y_i) = \sum_{j \in M} q_j y_{ij}$ for any bundle $y_i$. Consider an arbitrary feasible allocation $\y$. Since $(\x, p)$ is a WE, we have $x_i \in D_i(p)$, so $v_i(x_i) - p(x_i) \ge v_i(y_i) - p(y_i)$. Therefore
\begin{align*}
v_i(x_i) - \sum_{j \in M} q_j x_{ij} \ge&\ v_i(y_i) - \sum_{j \in M} q_j y_{ij}\\
\sum_{i \in N} v_i(x_i) - \sum_{i \in N} \sum_{j \in M} q_j x_{ij} \ge&\ \sum_{i \in N} v_i(y_i) - \sum_{i \in N} \sum_{j \in M} q_j x_{ij}\\
\sum_{i \in N} v_i(x_i) - \sum_{j \in M} q_j \sum_{i \in N} x_{ij} \ge&\ \sum_{i \in N} v_i(y_i) - \sum_{j \in M} q_j \sum_{i \in N} x_{ij}
\end{align*}
Furthermore, $\sum_{i \in N} x_{ij} = 1$ for all $j \in M$ with $q_j > 0$. Also, since $\y$ is a valid allocation, $\sum_{i \in N} y_{ij} \le 1$ for all $j \in M$. Therefore
\begin{align*}
\sum_{i \in N} v_i(x_i) - \sum_{j \in M} q_j \sum_{i \in N} x_{ij} \ge&\ \sum_{i \in N} v_i(y_i) - \sum_{j \in M} q_j \sum_{i \in N} x_{ij}\\
\sum_{i \in N} v_i(x_i) - \sum_{j \in M} q_j \ge&\ \sum_{i \in N} v_i(y_i) - \sum_{j \in M} q_j\\
\sum_{i \in N} v_i(x_i) \ge&\ \sum_{i \in N} v_i(y_i)
\end{align*}
Thus the utilitarian welfare of $\x$ is at least as high as that of any other allocation. We conclude that $\x$ maximizes utilitarian welfare.
\end{proof}

Note that no assumptions at all were made on the nature of the valuations: all we needed was $x_i \in \argmax_{y_i \in \X_i} (v_i(y_i) - p(y_i))$, and $\sum_{j \in M} x_{ij} = 1$ whenever $q_j > 0$. The most natural cases are $\X_i = \bbrpos^m$ (divisible goods) and $\X_i = \{0,1\}^m$ (indivisible goods), but the result does hold more broadly.

\section{Omitted proofs}\label{sec:proofs}

%
 
 \thmEuler*
 
\begin{proof}
Fix an arbitrary $\B \in \bbrpos^m$ and let $g(\lambda) = f(\lambda\B)$. Since $f$ is differentiable, so is $g$, and its derivative is given by the multidimensional chain rule: $\frac{\dif g(\lambda)}{\dif \lambda} = \sum_{j = 1}^m b_j \frac{\partial f(\lambda\B)}{\partial b_j}$. Since $f$ is homogeneous of degree $r$, we have $f(\lambda\B) = \lambda^r f(\B)$ for all $\lambda \ge 0$. Thus $g(\lambda) = \lambda^r f(\B)$ for all $\lambda \ge 0$, so we can differentiable both sides of this equation to get $\sum_{j = 1}^m b_j \frac{\partial f(\lambda\B)}{\partial b_j} = r\lambda^{r-1} f(\B)$. This holds for all $\lambda \ge 0$, so setting $\lambda = 1$ completes the proof.
\end{proof}

\lemHomoOneGood*
\begin{proof}
By Euler's Theorem (Theorem~\ref{thm:euler}), we have $x \mfrac{\dif f(x)}{\dif x} = r f(x)$ for all $x \in \bbrpos$. Let $y = f(x)$. We can solve this differential equation explicitly:
\begin{align*}
\frac{1}{y}\cdot \frac{\dif y}{\dif x} =&\ \frac{r}{x}\\
\int \frac{1}{y}\cdot \frac{\dif y}{\dif x} \dif x =&\ \int \frac{r}{x} \dif x\\
\int \frac{1}{y} \dif y =&\ r \int \frac{1}{x} \dif x\\
\ln y =&\ r \ln x + \ln c
\end{align*}
where $c$ (and thus $\ln c$) is some constant. Therefore
\begin{align*}
e^y =&\ e^{r \ln x + \ln c}\\
y =&\ c x^r
\end{align*}
Thus $f(x) = c x^r$, as required.
\end{proof}

\lemOptOneGood*

\begin{proof}
As in Section~\ref{sec:main}, strong duality for Program~\ref{pro:ces} implies that any optimal $\x$ must satisfy the KKT conditions. Thus $\x \in \Psi(\rho)$ if and only if there exists $q \in \bbrpos$ such that (1) stationarity holds: $\mfrac{\partial v_i(x_i)}{\partial x_i} v_i(x_i)^{\rho -1} \le q$ for all $i \in N$, and if $x_i > 0$, the inequality holds with equality, and (2) complementary slackness holds: either $\sum_{i \in N} x_i = 1$, or $q = 0$. 

Since we assume that $w_i > 0$ for all $i \in N$, any allocation with $\sum_{i \in N} x_i < 1$ is not Pareto optimal, and thus cannot be optimal for Program~\ref{pro:ces}. In other words, we must have $q > 0$. Thus complementary slackness simply requires that $\sum_{i \in N} x_i = 1$, and we can focus on stationarity.

We first handle $\rho = 1$. In this case, $\mfrac{\partial v_i(x_i)}{\partial x_i} v_i(x_i)^{\rho -1} = \mfrac{\partial v_i(x_i)}{\partial x_i} = w_i$. Thus if $\x \in \Psi(\rho)$ we must have $w_i \le q$, and if $x_i > 0$, then $w_i = q$. This implies that $q = \max_{k \in N} w_k$. Thus if $x_i > 0$, then $w_i = \max_{k \in N} w_k$, as required. 

For the rest of the proof, we assume $r\rho \ne 1$. Since $r,\rho \in (0,1]$, we have $0 < r\rho < 1$. By the definition of $v_i$, for an arbitrary allocation $\x$ and $i \in N$ we have
\[
\frac{\partial v_i(x_i)}{\partial x_i} v_i(x_i)^{\rho -1} = (w_i r x_i^{r-1})(w_i^{\rho-1} x_i^{r(\rho-1)}) = r w_i^\rho x_i^{r\rho - 1}
\]
Thus $q \ge rw_i^\rho x_i^{r\rho - 1}$. Since $r\rho < 1$, if $x_i = 0$, then $x_i^{r\rho-1}$ is undefined. Therefore stationarity is satisfied if and only if $q = r w_i^\rho x_i^{r\rho-1}$ for all $i \in N$, which is equivalent to
\begin{align}
x_i = (q/r)^{\frac{1}{r\rho-1}} w_i^{\frac{\rho}{1-r\rho}} \label{eq:x-expr}
\end{align}
Furthermore, if $\x$ satisfies Equation~\ref{eq:x-expr} for all $i \in N$, then $\sum_{i \in N} x_i = 1$ is equivalent to
\begin{align*}
\sum_{i \in N} (q/r)^{\frac{1}{r\rho-1}} w_i^{\frac{\rho}{1-r\rho}} =&\ 1\\
(q/r)^{\frac{1}{r\rho-1}} =&\ \Big(\sum_{i \in N} {w_i}^{\frac{\rho}{1-r\rho}}\Big)^{-1}
\end{align*}
Therefore $\x$ satisfies stationarity and complementary slackness (and thus satisfies $\x \in \Psi(\rho)$) if and only if
\[
x_i = \frac{{w_i}^{\frac{\rho}{1-r\rho}}}{\sum_{k \in N} {w_k}^{\frac{\rho}{1-r\rho}}}
\]
as required.
\end{proof}

\section*{Acknowledgements}

This research was supported in part by NSF grant CCF-1637418, 
ONR grant N00014-15-1-2786, and the NSF Graduate Research Fellowship under grant DGE-1656518.

\end{document}